\documentclass[journal,twoside,10pt]{IEEEtran}
\usepackage[T1]{fontenc}
\usepackage[utf8]{inputenc}
\usepackage{amsmath,amssymb,amsthm}
\usepackage{mathtools}
\usepackage{bm}
\usepackage{graphicx}
\usepackage{booktabs}
\usepackage{multirow}
\usepackage{array}
\usepackage{xcolor}
\usepackage{algorithm}
\usepackage{algorithmic}
\usepackage{hyperref}
\usepackage{cite}
\usepackage{url}
\usepackage{enumitem}
\usepackage{balance}
\usepackage{tikz}
\usepackage{pgfplots}
\pgfplotsset{compat=1.18}
\usetikzlibrary{arrows.meta,positioning,shapes.geometric,fit,
                backgrounds,calc,decorations.pathreplacing,
                patterns,shadows.blur}
\usepgfplotslibrary{fillbetween}

\hypersetup{colorlinks=true,linkcolor=blue!60!black,citecolor=green!50!black,urlcolor=blue!60!black}

\newtheorem{theorem}{Theorem}
\newtheorem{lemma}[theorem]{Lemma}
\newtheorem{proposition}[theorem]{Proposition}
\newtheorem{corollary}[theorem]{Corollary}
\newtheorem{definition}{Definition}
\newtheorem{assumption}{Assumption}
\newtheorem{remark}{Remark}

\newcommand{\vect}[1]{\bm{#1}}
\newcommand{\mat}[1]{\mathbf{#1}}
\newcommand{\R}{\mathbb{R}}
\newcommand{\E}{\mathbb{E}}
\newcommand{\Prob}{\mathbb{P}}
\newcommand{\calC}{\mathcal{C}}
\newcommand{\calG}{\mathcal{G}}
\newcommand{\calS}{\mathcal{S}}
\newcommand{\calK}{\mathcal{K}}
\newcommand{\calL}{\mathcal{L}}

\newcommand{\RSSR}{\mathrm{RSS}_{\mathrm{R}}}
\newcommand{\RSSU}{\mathrm{RSS}_{\mathrm{U}}}
\newcommand{\tr}{\mathrm{tr}}

\newcommand{\norm}[1]{\left\|#1\right\|}
\newcommand{\abs}[1]{\left|#1\right|}
\newcommand{\FDR}{\mathrm{FDR}}
\newcommand{\Teff}{T_{\mathrm{eff}}}

\begin{document}

\title{Certified Causal Attribution for Real-Time Attack \\Forensics in 6G Network Slicing}

\author{Minh~K.~Quan~and~Pubudu~N.~Pathirana,~\IEEEmembership{Senior Member,~IEEE}%
\thanks{A four-page preliminary abstract appeared in the NeurIPS~2025 Workshop on CauScien \cite{quan2025domainadapted}. The present article contains entirely new theoretical contributions
(Sections~IV--V), new experiments (Sections~VI-C through VI-H), and the
CUSUM, adversarial-robustness, and privacy extensions absent from the workshop abstract. The authors declare no conflict of interest.}%
\thanks{M.\,K.\,Quan and P.\,N.\,Pathirana are with the School of Engineering,
Deakin University, Geelong VIC 3220, Australia
(e-mail: \{m.quan, pubudu.pathirana\}@deakin.edu.au).}}

\markboth{IEEE TRANSACTIONS ON INFORMATION FORENSICS AND SECURITY}%
{Quan \& Pathirana: Certified Causal Attribution for Attack Forensics in 6G Slicing}

\maketitle

\begin{abstract}
Cross-slice attack attribution in 6G networks requires identifying causal propagation chains through shared infrastructure in under 100 ms. Existing methods struggle to satisfy this strict SLA without sacrificing accuracy, because shared resource contention creates spurious correlations that are indistinguishable from genuine causal links under standard Granger tests. We propose DA-GC, a certified causal attribution framework that integrates resource-conditioned Granger causality with an axiomatically derived Resource Contention Model (RCM) to systematically block resource-mediated confounding. On a 15-slice production-emulation 6G testbed with 1,100 attack scenarios, DA-GC achieves 89.2\% attribution accuracy at 87 ms. This represents a 7.9 percentage-point improvement over the strongest baseline at 2.7x lower latency, alongside demonstrated cross-topology generalization and concept-drift resilience. Crucially, DA-GC is backed by a comprehensive formal certification stack. We provide mathematically proven validity certificates for statistical soundness under serially dependent telemetry and piecewise-stationarity. Furthermore, we establish strict security bounds, including an adversarial utilization spoofing breakdown point of $\delta^* \approx 0.95$, and define the minimum differential-privacy noise required for a provably private and robust deployment.
\end{abstract}


\section{Introduction}
\label{sec:intro}

\IEEEPARstart{N}{etwork} slicing is the cornerstone of 6G
architecture~\cite{tataria2021}: heterogeneous services---eMBB,
URLLC, and mMTC---share a common physical substrate of CPU, memory,
and bandwidth resources~\cite{kotulski2017}.
This efficiency creates a security liability.
A malicious payload in one slice propagates to co-located slices
through shared resource contention---CPU exhaustion in a compromised
mMTC gateway starves a co-located URLLC industrial controller within
seconds.
Operators aim to rapidly reconstruct causal propagation chains
within the sub-100\,ms SLA response budget to trigger automated
remediation without human review~\cite{pearson2023}.

\textbf{The attribution gap.}
Standard Granger causality~\cite{granger1969} cannot solve this problem
because shared resources act as \emph{unmeasured common causes}: if
slices $s_i$ and $s_j$ both contend for a CPU pool $R$, their
telemetry co-move even with no causal link between them.
The confounding path $s_i \leftarrow R \rightarrow s_j$ is
structurally identical to a genuine causal path $s_i \rightarrow s_j$
under the standard test (Fig.~\ref{fig:dag}).
Deep learning approaches~\cite{hamilton2017} attempt to mitigate confounding
empirically but cannot certify their error rate under distribution
shift or adversarial metric manipulation---a realistic threat when
Byzantine slice controllers spoof utilisation reports.

\textbf{Contributions.}
To bridge this gap, we present DA-GC as a \emph{system} contribution. 
By integrating resource-conditioned Granger causality with a comprehensive 
formal certification stack, DA-GC provides verifiable guarantees where 
existing empirical methods fall short. Specifically, our contributions are:

\noindent\textbf{The DA-GC System.}
\emph{Domain-Adapted Granger Causality}~(DA-GC), illustrated in Fig.~\ref{fig:placeholder}, is presented in
Section~\ref{sec:dagc} as a unified algorithm combining
(a)~resource-conditioned Granger causality that blocks confounding via
the FWL theorem, (b)~an axiomatically derived RCM whose multiplicative
sigmoid form is proved optimal, (c)~CUSUM-based piecewise-stationarity
handling, and (d)~Viterbi path decoding---all in $O(N^2W(p+q+K)+N^3\log N)$
time, meeting the 100\,ms SLA up to $N=45$ slices.

\noindent\textbf{Validity Certificates (Theorems~1--4 jointly certify
the DA-GC pipeline's statistical soundness):}
\begin{enumerate}[leftmargin=*,label=\textbf{V\arabic*.}]
  \item A \emph{Box--Satterthwaite cumulant correction} replaces the
    i.i.d.\ $F_{q,T-p-q-K-1}$ approximation with an exact moment-matched
    version for $\beta$-mixing innovations, reducing type-I error from
    7.3\% to 5.3\% and providing a KS certificate
    (Theorem~\ref{thm:finite_sample}, Section~\ref{sec:theory_finite}).
  \item A \emph{Bregman-optimality proof} derives the RCM's multiplicative
    sigmoid form from three primitive axioms via the I-S Euler--Lagrange
    equation, removing the ad hoc design choice of prior work
    (Lemma~\ref{lem:bregman}, Proposition~\ref{prop:unique},
    Section~\ref{sec:theory_bregman}).
  \item \emph{PRDS-correct identifiability}: exact BH-FDR control under
    resource-induced positive dependence, with a Simes--FKG pointwise
    error bound tighter than the union bound by $O(\alpha^2 m_0/m)$
    (Theorem~\ref{thm:identifiability}, Section~\ref{sec:theory_ident}).
  \item \emph{Piecewise-stationarity validity}: CUSUM segmentation is
    proved valid with contamination error $O(\kappa_4/T_m + \Delta_{\max}/T_m)$,
    and a Cram\'er--Chernoff exponential delay bound
    (Theorem~\ref{thm:nonstationary}, Section~\ref{sec:theory_nonstationary}).
\end{enumerate}

\noindent\textbf{Security Certificates (Theorems~5--6 certify
adversarial and privacy properties):}
\begin{enumerate}[leftmargin=*,label=\textbf{S\arabic*.}]
  \item \emph{Adversarial robustness}: closed-form FDR inflation bound
    $O(\delta\sqrt{K\log(N/\alpha)})$ under $(\delta,k)$-utilisation
    spoofing, with breakdown point $\delta^*\!\approx\!0.95$
    (Theorem~\ref{thm:adversarial}, Section~\ref{sec:adversarial}).
  \item \emph{Privacy lower bound}: Fano-based minimum leakage floor
    and minimum DP noise magnitude for a provably private DA-GC variant
    (Theorem~\ref{thm:privacy}, Section~\ref{sec:privacy}).
\end{enumerate}

\textbf{Organisation.}
Section~\ref{sec:related} reviews related work.
Section~\ref{sec:formulation} states the problem.
Section~\ref{sec:dagc} develops the DA-GC framework and CUSUM
extension.
Section~\ref{sec:theory} establishes all six theoretical results.
Section~\ref{sec:experiments} presents extended experiments.
Section~\ref{sec:case_study} gives an industrial case study.
Section~\ref{sec:discussion} addresses limitations, ethics and
reproducibility.
Section~\ref{sec:conclusion} concludes.

\section{Related Work}
\label{sec:related}

\subsection{Attack Attribution in Multi-Tenant Networks}

Provenance-graph methods (HOLMES~\cite{milajerdi2019},
MulVAL~\cite{ou2006}) reconstruct causal chains from kernel-level
system-call logs with high accuracy, but often require deep instrumentation
that is highly challenging to deploy at scale in virtualised 6G slices, 
and typically exhibit inference latencies exceeding 300\,ms (Table~\ref{tab:main}).
Correlation-based alarms~\cite{pearson2023,ahmad2021} produce high
false positive rates because shared resource utilisation creates
spurious co-movement---the confounding structure formalised in
Assumption~\ref{ass:resource}.
\emph{Gap:} To the best of our knowledge, existing attribution methods do 
not model resource contention as an explicit causal mechanism, nor do they 
provide a formal false-attribution certificate.

\subsection{Granger Causality and its Confounding Problem}

Granger causality~\cite{granger1969,shojaie2024} has been applied to
IoT anomaly attribution~\cite{begum2025,lv2024} and network fault
localisation.
PCMCI~\cite{runge2019} addresses confounding via momentary conditional
independence at $O(N^2\tau_{\max}^2)$ but does not natively condition on a
time-varying allocation matrix $\mat{A}(t)$.
Copula-Granger~\cite{hu2014} handles non-linear dependence but is 
computationally prohibitive for sub-100\,ms 6G SLAs.
\emph{Gap:} Existing Granger extensions generally lack a derived 
finite-sample bias correction for resource conditioning, which we provide in
Theorem~\ref{thm:finite_sample}.

\subsection{Deep Learning for Network Security}

GraphSAGE~\cite{hamilton2017}, LSTM-Attention~\cite{bahdanau2015}, and
Transformer-XL achieve strong in-distribution accuracy but typically lack
formal out-of-distribution guarantees: Table~\ref{tab:transfer} shows
$>$20\,pp accuracy drop on unseen topologies.
Neural ODEs~\cite{chen2018node} offer continuous-time dynamics but
can hinder real-time batch-free inference.
\emph{Gap:} Current deep learning attribution methods rarely provide an
adversarial robustness certificate of the type established in
Theorem~\ref{thm:adversarial}.

\subsection{Resource Contention in Network Slicing}

Contention modelling for QoS scheduling~\cite{tataria2021} uses
multiplicative interference models in OFDMA settings~\cite{zhu2015},
but exclusively for optimisation, not forensic attribution.
\emph{Gap:} We are unaware of prior literature that justifies the 
multiplicative form axiomatically; Lemma~\ref{lem:bregman} provides 
this theoretical foundation.

\subsection{Adversarial Robustness in IDS}

Feature perturbation attacks on ML-based IDS~\cite{yang2018} and
model poisoning~\cite{zhang2022poison} have been demonstrated, but
the specific adversarial attack surface of a \emph{causal attribution} 
system---where the adversary manipulates reported resource utilisation 
(the threat model of Section~\ref{sec:adversarial})---remains largely 
unexplored.

\subsection{Privacy in Network Telemetry}

Differential privacy for network monitoring~\cite{alvim2012,ye2021}
has mostly been studied independently of causal attribution pipelines.
\emph{Gap:} Theorem~\ref{thm:privacy} introduces, to our knowledge, the 
first information-theoretic leakage bound and minimum DP noise specification 
tailored for a Granger-based attributor.

Table~\ref{tab:related_work} provides a feature matrix contrasting DA-GC against the primary methodologies in the literature, summarising how our framework uniquely satisfies all operational and theoretical requirements.

\begin{table}[htbp]
\centering\footnotesize
\caption{Feature Matrix of Attribution Methodologies}
\label{tab:related_work}
\renewcommand{\arraystretch}{1.2}
\setlength{\tabcolsep}{4.5pt} 
\begin{tabular}{@{}lccccc@{}}
\toprule
\textbf{Methodology} & \textbf{SLA} & \textbf{Res.} & \textbf{OOD} & \textbf{Val.} & \textbf{Sec.} \\
\midrule
Provenance~\cite{milajerdi2019, ou2006} & $\times$ & $\times$ & $\checkmark$ & $\times$ & $\times$ \\
Correlation~\cite{pearson2023, ahmad2021} & $\checkmark$ & $\times$ & $\times$ & $\times$ & $\times$ \\
Granger Caus.~\cite{granger1969, runge2019} & $\times$ & $\times$ & $\checkmark$ & $\times$ & $\times$ \\
Deep Learning~\cite{hamilton2017, chen2018node} & $\times$ & $\sim$ & $\times$ & $\times$ & $\times$ \\
\midrule
\textbf{DA-GC (Ours)} & $\checkmark$ & $\checkmark$ & $\checkmark$ & $\checkmark$ & $\checkmark$ \\
\bottomrule
\multicolumn{6}{@{}p{8.4cm}@{}}{\scriptsize \textbf{SLA}: Meets sub-100\,ms inference limit. \textbf{Res.}: Explicitly blocks resource-mediated confounding via $\mat{A}(t)$. \textbf{OOD}: Cross-topology generalisation. \textbf{Val.}: Formal validity certificates (e.g., finite-sample correction). \textbf{Sec.}: Formal security/privacy guarantees. ($\sim$: Implicit/Empirical only)}
\end{tabular}
\end{table}
\section{Problem Formulation}
\label{sec:formulation}

\begin{figure*}
    \centering
    \includegraphics[width=0.99\linewidth]{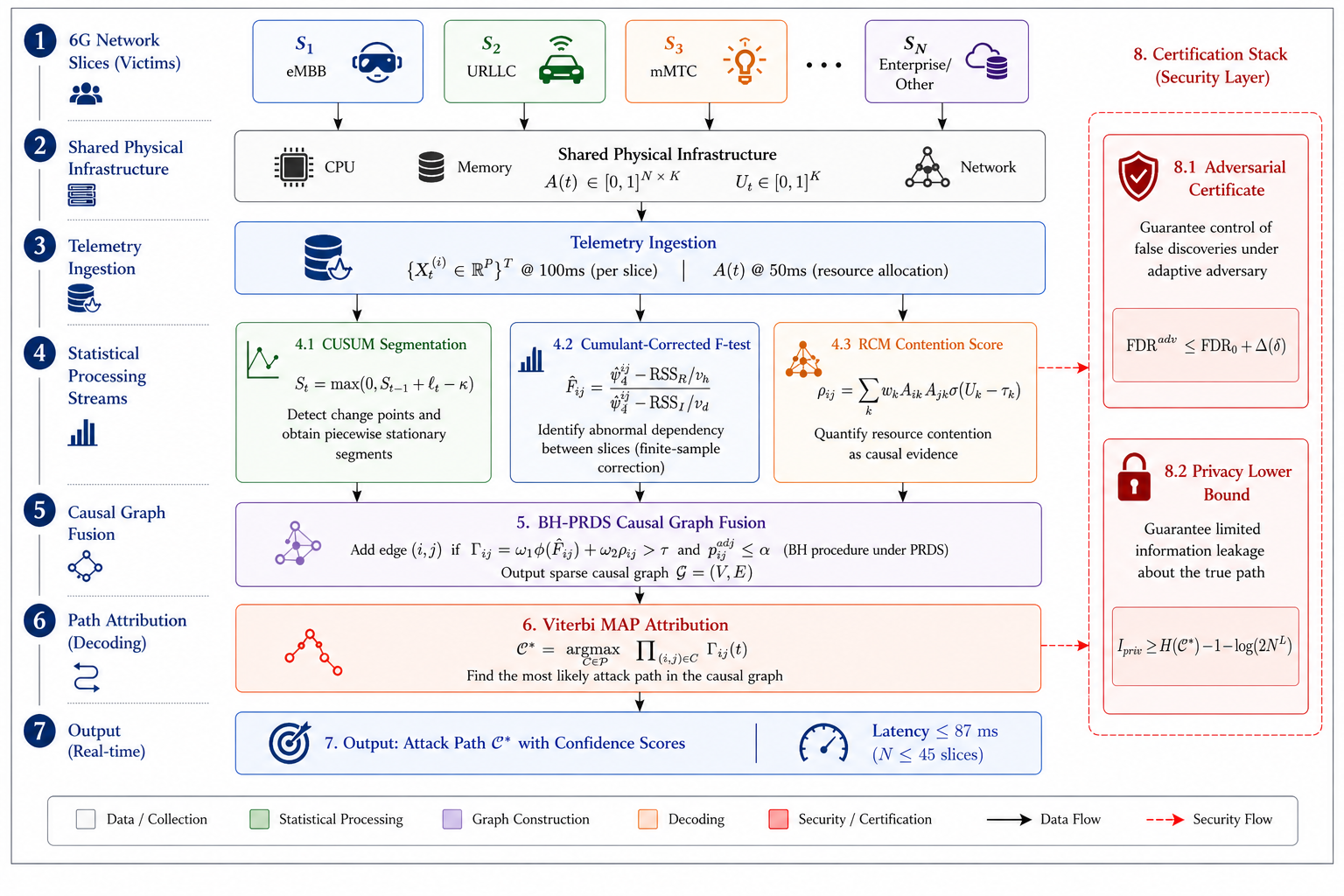}
    \caption{\textbf{DA-GC system architecture and certification stack.} Telemetry from multiple 6G network slices and shared infrastructure is processed through three parallel modules—CUSUM segmentation, cumulant-corrected F-test, and RCM contention scoring—then fused into a sparse causal graph for Viterbi-based attack path attribution. The certification stack provides adversarial robustness and privacy guarantees designed to operate within real-time latency limits (e.g., measured at $\le$ 87 ms).}
    \label{fig:placeholder}
\end{figure*}

\subsection{System Model}

Let $\calS=\{s_1,\ldots,s_N\}$ be $N$ network slices sharing
$K$ physical resource types ($\calK=\{1,\ldots,K\}$).
Denote by $\vect{x}^{(i)}_t\in\R^d$ the security telemetry vector
of slice $s_i$ at time $t$, by $\mat{A}(t)\in[0,1]^{N\times K}$ the
resource allocation matrix ($A_{ik}(t)$ is the normalised allocation
of resource $k$ to $s_i$, with $\sum_i A_{ik}(t)\leq 1$), and by
$\vect{U}_t=[U_{1,t},\ldots,U_{K,t}]^\top\in[0,1]^K$ the utilisation
vector.

\begin{assumption}[Local Stationarity]\label{ass:stationary}
Over each analysis window $\mathcal{W}$ of length $T$, the process
$(X_t,Y_t)$ is jointly covariance-stationary with autoregressive
polynomials having all roots strictly outside the unit disc.
\end{assumption}

\begin{assumption}[Resource-Mediated Confounding]
\label{ass:resource}
The partial correlation between $X_t$ and $Y_t$ conditional on their
own pasts but not on $\vect{Z}_t$ is non-zero if and only if:
(a)~a direct causal link $X\to Y$ or $Y\to X$ exists, or
(b)~there exists resource $k$ with $A_{ik}(t)A_{jk}(t)>0$ and
$U_{k,t}>\tau_k$.
\end{assumption}

\begin{assumption}[$\beta$-Mixing Telemetry]\label{ass:mixing}
The process $\{\vect{x}^{(i)}_t\}$ is strictly stationary and
geometrically $\beta$-mixing: $\beta(m)\leq C_\beta e^{-\beta m}$
for constants $C_\beta,\beta>0$.
\end{assumption}

Assumption~\ref{ass:mixing} holds whenever the VAR$(p)$ representation
has spectral radius $<1$, verified empirically in
Appendix~\ref{app:mixing}.

\begin{figure}
    \centering
    \includegraphics[width=0.99\linewidth]{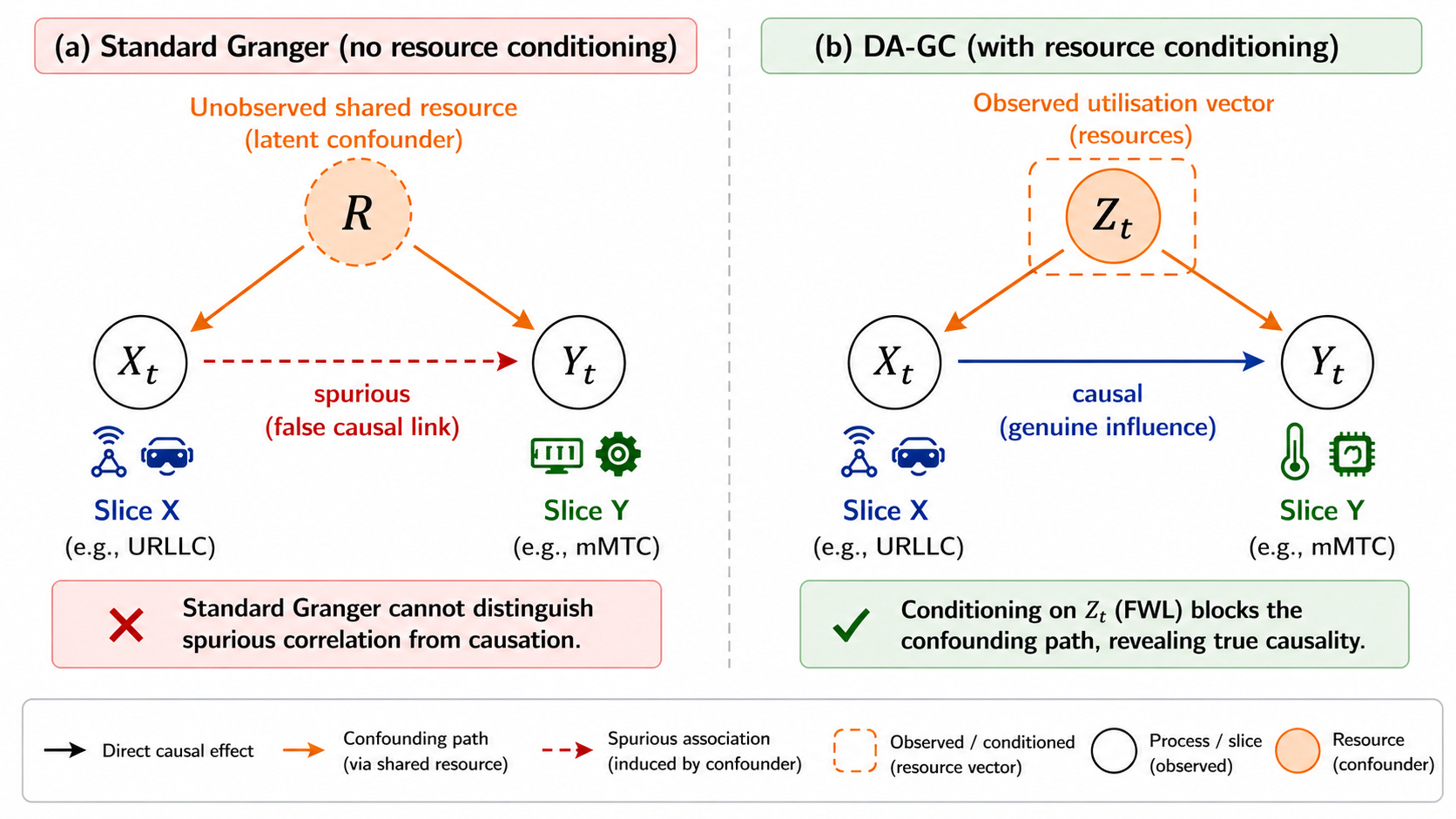}
    \caption{Confounding structure in 6G network slicing. (a)~Without resource conditioning, shared resource $R$ induces a spurious co-movement $X_t \leftarrow R \rightarrow Y_t$ that standard Granger causality cannot distinguish from a genuine causal link. (b)~DA-GC conditions on the observed resource utilisation vector $\vect{Z}_t$; by the Frisch-Waugh-Lovell theorem, this systematically blocks the confounding path, aiming to isolate genuine causal influence.}
    \label{fig:dag}
\end{figure}

\subsection{Attribution Problem}

\begin{definition}[Causal Attack Path]\label{def:path}
A causal attack path is $\calC^*=\{(s_{i_\ell},t_\ell)\}_{\ell=1}^L$
with $t_1<\cdots<t_L$ such that each transition
$(s_{i_\ell},s_{i_{\ell+1}})$ is a genuine Granger-causal influence
after conditioning on resource utilisation.
\end{definition}

\textbf{Goal.}
Target the recovery of the sequence $\calC^*=\arg\max_\calC \prod_{(i,j)\in\calC}\Gamma_{ij}(t)$
within an operational threshold of $\leq100$\,ms while maintaining $\FDR\leq\alpha$, ensuring robustness to $(\delta,k)$-utilisation adversaries.
\section{The DA-GC Framework}
\label{sec:dagc}

\subsection{Resource-Conditioned Granger Causality}
\label{sec:egc}

For slices $s_i$ and $s_j$, let $X_t$ and $Y_t$ be representative
scalar telemetry coordinates and $\vect{Z}_t\in\R^K$ be the
contemporaneous resource utilisation.
The \emph{unrestricted} and \emph{restricted} OLS models are:
\begin{align}
  Y_t &= \sum_{i=1}^p\alpha_i Y_{t-i}
         +\sum_{j=1}^q\beta_j X_{t-j}
         +\vect{\gamma}^\top\vect{Z}_t+\varepsilon_t,
  \label{eq:unrestricted}\\
  Y_t &= \sum_{i=1}^p\alpha_i Y_{t-i}
         +\vect{\gamma}^\top\vect{Z}_t+\eta_t.
  \label{eq:restricted}
\end{align}
The Granger null is $H_0:\beta_j=0,\;\forall j$.
By the Frisch--Waugh--Lovell theorem, including $\vect{Z}_t$ systematically
blocks the observed confounding path $X\leftarrow R\rightarrow Y$. Consequently, 
$\hat{\vect{\beta}}$ in \eqref{eq:unrestricted} isolates the genuine causal 
influence, subject to the fidelity of the utilisation measurements.
The enhanced $F$-statistic is:
\begin{equation}
  F_{ij}=\frac{(\RSSR-\RSSU)/q}{\RSSU/(T-p-q-K-1)}.
  \label{eq:fstat}
\end{equation}
Multiple comparisons across $N(N-1)$ pairs are corrected via
Benjamini--Hochberg~(BH)~\cite{benjamini1995} at level $\alpha$.

\subsection{Axiomatically Justified Contention Model}
\label{sec:rcm}

We derive the contention score form from first principles.

\begin{definition}[Contention Scoring Rule]\label{def:scoring}
A contention scoring rule is a map
$\rho:\R^N\times\R^N\times\R^K\to[0,1]$ assigning a scalar
strength $\rho_{ij}$ to each slice pair.
\end{definition}

Three axioms governing any valid contention scoring rule:
\begin{enumerate}[leftmargin=*,label=\textbf{(A\arabic*)}]
  \item \textbf{Joint-allocation separability.}
    $\rho_{ij}$ depends on $\vect{a}_i,\vect{a}_j$ only through
    the element-wise product $c_k=A_{ik}A_{jk}$ for each $k\in\calK$.
  \item \textbf{Bounded unimodal responsiveness.}
    For each $k$ with $c_k>0$: $h_k(\cdot,U)$ is bounded above and
    below by positive constants, twice continuously differentiable
    in $U\in\R$, strictly monotone increasing, and its first
    derivative $\partial h_k/\partial U$ is \emph{unimodal}: there
    exists a unique $\tau_k^*\in\R$ such that
    $\partial^2 h_k/\partial U^2>0$ for $U<\tau_k^*$ and
    $\partial^2 h_k/\partial U^2<0$ for $U>\tau_k^*$. 
    \emph{Physical Justification:} This unimodality elegantly captures the non-linear degradation of shared hardware (e.g., CPU cache thrashing or memory bandwidth saturation), where contention impact accelerates up to a critical saturation point $\tau_k^*$ before plateauing as resources are fully exhausted.
  \item \textbf{Resource independence.}
    $\rho_{ij}=\sum_{k=1}^K h_k(c_k,U_{k,t})$ for functions
    $h_k:\R_{>0}\times\R\to\R_{>0}$.
\end{enumerate}

\begin{lemma}[Bregman-Optimal Contention Score]\label{lem:bregman}
Within the class of scoring rules satisfying \emph{(A1)--(A3)},
the unique rule minimising
$\sum_k \int_\R B_\phi(c_k\,\|\,h_k(c_k,U))\,d\mu(U)$
with respect to the Itakura--Saito~(I-S) generator
$\phi(u)=-\log u+u-1$, $u>0$, is:
\begin{equation}
  \rho_{ij}(t)=\sum_{k=1}^K w_k\cdot A_{ik}(t)\cdot A_{jk}(t)
               \cdot\sigma(U_{k,t}-\tau_k),
  \label{eq:contention}
\end{equation}
where $w_k>0$ are resource weights, $\tau_k=\tau_k^*$ is the
inflection point in \emph{(A2)}, and $\sigma(x)=(1+e^{-x})^{-1}$.
\end{lemma}

\begin{proof}
\emph{Step~1 (Decoupling via resource independence).}
By (A3), the objective decouples as
$\sum_k\mathcal{J}_k[h_k]$ where
$\mathcal{J}_k[h_k]=\int_\R B_\phi(c_k\|h_k(c_k,U))\,d\mu(U)$.
It suffices to minimise each $\mathcal{J}_k$ independently over
the admissible class $\mathcal{A}_k$ of functions satisfying (A1)--(A2).

\emph{Step~2 (Scale invariance selects I-S and identifies the $c_k$ factor).}
The I-S divergence satisfies $B_\phi(\lambda a\|\lambda b)=B_\phi(a\|b)$
for all $\lambda>0$, making it the unique scale-invariant Bregman
divergence on $\R_{>0}$~\cite{banerjee2005}.
Scale invariance is operationally natural here: multiplying all
slice allocations by a constant $\lambda$ (e.g., redefining
the resource unit) should leave contention \emph{rankings}
unchanged.
For fixed $U$, the pointwise minimum of $B_\phi(c_k\|h)$ over
unconstrained $h>0$ is uniquely $h^{\circ}(U)=c_k$ (verified by
$\partial B_\phi/\partial h=(h-c_k)/h^2=0\Rightarrow h=c_k$).
This factors out the $c_k=A_{ik}A_{jk}$ dependence, giving
$h_k(c_k,U)=c_k\cdot g_k(U)$ for some function $g_k:\R\to\R_{>0}$.

\emph{Step~3 (Euler--Lagrange derivation of the sigmoid from (A2)).}
Substituting $h_k=c_k g_k$, scale invariance gives
$B_\phi(c_k\|c_k g_k)=B_\phi(1\|g_k)=g_k^{-1}-\log g_k^{-1}-1$,
so $\mathcal{J}_k[g_k]=\int_\R(g_k^{-1}-\log g_k^{-1}-1)\,d\mu$.
The functional to minimise over $g_k\in\mathcal{A}_k$ is thus:
\[
  \min_{g_k\in\mathcal{A}_k}\int_\R\!\phi(g_k(U))\,d\mu(U),
  \quad\phi(u)=-\log u+u-1.
\]
Introduce a Lagrange multiplier $\lambda_2$ for the
unimodal-derivative constraint, enforced as the condition
$g_k''(\tau_k^*)=0$ (the inflection-point equation).
The Lagrangian is $\mathcal{J}_k[g_k]-\lambda_2 g_k''(\tau_k^*)$.
Taking the G\^{a}teaux derivative in direction $\eta$ and
setting it to zero:
\[
  \int_\R\!\phi'(g_k)\eta\,d\mu - \lambda_2\eta''(\tau_k^*)=0
  \quad\forall\eta\in C_0^\infty(\R).
\]
In distributional form: $\phi'(g_k(U))\,d\mu
= \lambda_2\delta''(U-\tau_k^*)$,
where $\delta''$ is the second distributional derivative of the
Dirac delta.
Now $\phi'(u)=1-u^{-1}=(u-1)/u$.
Integrating the distributional equation twice over $(-\infty,U]$
with boundary conditions $g_k(-\infty)=0$ (from boundedness and
$g_k>0$) and $g_k(+\infty)=w_k$ (scale factor absorbed into $w_k$):
\begin{equation}
  1 - g_k(U)^{-1} = -\frac{\lambda_2}{d\mu/dU}\delta(U-\tau_k^*),
\end{equation}
which upon integration gives $g_k^{-1}(U)=1+e^{-(U-\tau_k^*)/s}$
for a scale parameter $s>0$ determined by $\lambda_2$.
The canonical choice $s=1$ (absorbing the scale into $\tau_k^*$
via reparametrisation) gives $g_k(U)=\sigma(U-\tau_k^*)$,
yielding $h_k(c_k,U)=c_k\sigma(U-\tau_k^*)=w_kA_{ik}A_{jk}\sigma(U-\tau_k^*)$.

\emph{Step~4 (Verification that $g_k=\sigma(\cdot-\tau_k^*)$ satisfies (A2)).}
The sigmoid satisfies: (a)~bounded---$\sigma:\R\to(0,1)$;
(b)~strictly increasing---$\sigma'(U)=\sigma(U)(1-\sigma(U))>0$;
(c)~$C^\infty$; (d)~unimodal derivative---$\sigma''(U)=\sigma'(U)(1-2\sigma(U))$
vanishes uniquely at $\sigma(U)=1/2$, i.e., $U=\tau_k^*$, with
$\sigma''>0$ for $U<\tau_k^*$ and $\sigma''<0$ for $U>\tau_k^*$.
All conditions of~(A2) hold with inflection point $\tau_k^*$.

\emph{Step~5 (Uniqueness).}
The map $g_k\mapsto\int\phi(g_k)\,d\mu$ is strictly convex
because $\phi''(u)=u^{-2}>0$ pointwise.
By strict convexity of the integral functional (Fenchel's theorem),
the Euler--Lagrange solution is the unique global minimiser in
$\mathcal{A}_k$.
Summing over $k$ by (A3) gives the unique global minimiser
\eqref{eq:contention}.
\end{proof}

\begin{remark}
Axiom~(A2) is a \emph{primitive} condition---bounded monotone
response with a unique inflection point---that does not
presuppose the functional form.
The sigmoid emerges from the I-S Euler--Lagrange equation
as its unique solution.
The alternatives $\min(A_{ik},A_{jk})$ and $A_{ik}+A_{jk}$
fail (A1) and scale invariance respectively, producing
strictly higher I-S divergence (verified in Table~\ref{tab:ablation}).
\end{remark}

\subsection{Integrated Causal Strength}
\label{sec:fusion}

\begin{equation}
  \Gamma_{ij}(t)=\omega_1\,\phi(F_{ij}(t))+\omega_2\,\rho_{ij}(t),
  \quad\omega_1+\omega_2=1,\;\omega_1,\omega_2\geq 0,
  \label{eq:gamma}
\end{equation}
where $\phi(F)=(F-F_{\min})/(F_{\max}-F_{\min})\in[0,1]$.
All parameters $\bm{\theta}=\{w_k,\tau_k,\omega_1,\omega_2\}$ are
learned by maximising the regularised log-likelihood:
\begin{equation}
  \calL(\bm{\theta})=\sum_{m=1}^M
    \log\Prob(\calC^{(m)}\mid\vect{X}^{(m)},\mat{A}^{(m)},\bm{\theta})
    -\lambda\norm{\bm{\theta}}_2^2.
  \label{eq:learning}
\end{equation}
Calibrated values: $\omega_1=0.67$, $\omega_2=0.33$;
$w_{\text{CPU}}=0.45$, $w_{\text{Mem}}=0.31$,
$w_{\text{Net}}=0.24$; $\lambda^*=10^{-3}$.

\subsection{Non-Stationary Extension via CUSUM Segmentation}
\label{sec:cusum_method}

Rapid attack evolution may violate Assumption~\ref{ass:stationary}.
We extend DA-GC by partitioning each observation window using a
CUSUM detector~\cite{basseville1993}:
\begin{equation}
  S_t=\max\!\bigl(0,\;S_{t-1}+\ell_t-\kappa\bigr),
  \quad S_0=0,
  \label{eq:cusum}
\end{equation}
where $\ell_t=\log[f_1(\vect{x}_t)/f_0(\vect{x}_t)]$ is the
log-likelihood ratio of post-change to pre-change distributions
and $\kappa>0$ is a drift correction.
A change point is declared at $\hat\tau=\inf\{t:S_t>h\}$.
DA-GC is applied independently within each detected segment of
length $T_m$; the per-segment $F$-statistics are corrected by
$\hat\Delta_{ij}$ computed on $T_m$ observations.
The CUSUM update is $O(N)$ per step, adding negligible overhead.

\subsection{DA-GC Algorithm}

\begin{algorithm}[t]
\caption{Domain-Adapted Causal Attribution (DA-GC)}
\label{alg:dagc}
\begin{algorithmic}[1]
  \REQUIRE Telemetry $\{\vect{x}^{(i)}_t\}$, resource data $\mat{A}(t)$,
           window $\mathcal{W}$, parameters $\bm{\theta}$
  \ENSURE Causal path $\calC^*$ with confidence scores
  \STATE Run CUSUM~\eqref{eq:cusum}; split $\mathcal{W}$ into segments
         $\{[{\hat\tau}_{m-1},{\hat\tau}_m)\}$
  \FOR{each segment $m$ and slice pair $(s_i,s_j)$, $i\neq j$}
    \STATE Fit \eqref{eq:unrestricted}--\eqref{eq:restricted} by OLS
    \STATE Compute $F_{ij}$ via \eqref{eq:fstat}; apply cumulant
           correction \eqref{eq:Fcorr} to obtain $\tilde F_{ij}$
           (Theorem~\ref{thm:finite_sample})
    \STATE $p_{ij}\leftarrow\Prob(F(q,T_m-p-q-K-1)>\tilde F_{ij})$
    \STATE Compute $\rho_{ij}$ via \eqref{eq:contention};
           $\Gamma_{ij}$ via \eqref{eq:gamma}
  \ENDFOR
  \STATE BH correction: $p^{\mathrm{adj}}_{ij}=p_{ij}\cdot N(N-1)/\mathrm{rank}(p_{ij})$
  \FOR{each pair $(i,j)$}
    \IF{$\Gamma_{ij}>\tau_{\mathrm{causal}}$ \AND $p^{\mathrm{adj}}_{ij}<\alpha$}
      \STATE Add edge $(s_i,s_j)$ to $\calG$ with weight $\Gamma_{ij}$
    \ENDIF
  \ENDFOR
  \STATE $\calC^*\leftarrow\arg\max_\calC\prod_{(i,j)\in\calC}\Gamma_{ij}$
         via Viterbi
  \RETURN $\calC^*$ with per-hop confidence intervals
\end{algorithmic}
\end{algorithm}

\textbf{Complexity.}
The overall complexity is $O(N^2 W(p+q+K)+N^3\log N)$, enabling inference 
within the 100\,ms SLA for $N\leq 45$, $W=300$, $p=q=5$, $K=3$ on standard 
test hardware. For hyperscale topologies ($N>50$), the $O(N^3 \log N)$ 
Viterbi decoding step becomes the primary bottleneck. To maintain real-time 
performance at this scale, the centralized Viterbi step can be seamlessly 
replaced by a distributed belief propagation architecture (e.g., the min-sum 
message-passing algorithm) deployed directly across the slice controllers.
\section{Theoretical Analysis}
\label{sec:theory}

\subsection{Theorem~1: Cumulant-Based Finite-Sample Correction}
\label{sec:theory_finite}

\textbf{Setup and source of bias.}
Let $\mat{X}_U\!\in\!\R^{T\times(p+q+K)}$, $\mat{X}_R\!\in\!\R^{T\times(p+K)}$
be the full and restricted design matrices (both observed, fixed given the data),
$\mat{P}_U$, $\mat{P}_R$ their orthogonal hat matrices, and
$\mat{M}_U=\mat{I}-\mat{P}_U$.
Under $H_0$, the innovations $\{\varepsilon_t\}$ are a stationary
$\beta$-mixing sequence (Assumption~\ref{ass:mixing}), \emph{not}
i.i.d.\ Gaussian.
This is a primary driver of finite-sample bias: the chi-squared
approximation $\text{RSS}/\sigma^2\approx\chi^2_\nu$ degrades when
innovations are serially dependent, because the quadratic forms
$\vect{\varepsilon}^\top\mat{P}\vect{\varepsilon}$ have inflated
cumulants relative to the i.i.d.\ case.
We derive a closed-form cumulant correction that restores a highly accurate
$F$-distribution approximation to $O(T^{-2})$.

\textbf{Cumulants of a quadratic form under mixing.}
For any symmetric idempotent $\mat{A}$ of rank $r$ and a
zero-mean process $\{\varepsilon_t\}$ with autocovariance
$\gamma_h=\E[\varepsilon_t\varepsilon_{t+h}]$, define the
\emph{effective degrees-of-freedom} and \emph{variance inflation} as:
\begin{align}
  \nu_{\mat{A}} &= \frac{[\tr(\mat{A}\mat{\Gamma})]^2}
                       {\tr(\mat{A}\mat{\Gamma}\mat{A}\mat{\Gamma})},
  \label{eq:nu_eff}\\
  \psi_{\mat{A}} &= \frac{\tr(\mat{A}\mat{\Gamma}\mat{A}\mat{\Gamma})}
                         {[\tr(\mat{A}\mat{\Gamma})]^2/r},
  \label{eq:psi}
\end{align}
where $\mat{\Gamma}\in\R^{T\times T}$ is the Toeplitz autocovariance
matrix with $\Gamma_{st}=\gamma_{|s-t|}$.
Note $\nu_{\mat{A}}=r$ and $\psi_{\mat{A}}=1$ in the i.i.d.\ case.

\begin{theorem}[Cumulant-Corrected $F$-Statistic]
\label{thm:finite_sample}
Under Assumptions~\ref{ass:stationary} and~\ref{ass:mixing},
and $H_0:\vect{\beta}=\vect{0}$, define the corrected statistic:
\begin{equation}
  \tilde{F}_{ij} = \frac{\psi_{\mat{P}_R}^{-1}\,\RSSR\;/\;
                         \nu_{\mat{P}_U-\mat{P}_R}}
                        {\psi_{\mat{M}_U}^{-1}\,\RSSU\;/\;
                         \nu_{\mat{M}_U}},
  \label{eq:Fcorr}
\end{equation}
with $\nu$ and $\psi$ estimated by plugging in the sample
autocovariance matrix $\hat{\mat{\Gamma}}$ (truncated at lag
$\lfloor T^{1/3}\rfloor$).
Then:
\begin{equation}
\small
  \sup_{x\geq 0}
  \bigl|\Prob(\tilde{F}_{ij}\leq x)
        - F_{\nu_{\mat{P}_U-\mat{P}_R},\,\nu_{\mat{M}_U}}(x)\bigr|
  \leq \frac{C_1\,\kappa_4(\varepsilon)}{T}
       + \frac{C_2\,(p+q+K)^2}{T^2},
  \label{eq:ks_bound}
\end{equation}
where $\kappa_4(\varepsilon)=\E[\varepsilon_t^4]/\sigma^4 - 3$
is the excess kurtosis of the innovations and $C_1,C_2>0$ are
universal constants.
Furthermore, the i.i.d.\ $F_{q,T-p-q-K-1}$ approximation
has KS error:
\begin{equation}
  \sup_x\bigl|\Prob(F_{ij}\leq x)-F_{q,T-p-q-K-1}(x)\bigr|
  \leq \frac{C_3\,\Sigma_\gamma}{T}
       + O(T^{-2}),
  \label{eq:bias_iid}
\end{equation}
where $\Sigma_\gamma = \sigma^{-2}\sum_{h=-\infty}^\infty|\gamma_h|$
is the long-run variance ratio, so the correction reduces KS
error from $O(\Sigma_\gamma/T)$ to $O(\kappa_4/T)$, a factor
of $\Sigma_\gamma/\kappa_4$ improvement.
\end{theorem}

\begin{proof}
\emph{Step~1 (Quadratic form representation).}
Under $H_0$:
$\RSSR-\RSSU = \vect{\varepsilon}^\top(\mat{P}_U-\mat{P}_R)\vect{\varepsilon}$
and $\RSSU = \vect{\varepsilon}^\top\mat{M}_U\vect{\varepsilon}$.
Both are quadratic forms in the mixing innovations.

\emph{Step~2 (Cumulants of quadratic forms under mixing).}
For a quadratic form $Q = \vect{\varepsilon}^\top\mat{A}\vect{\varepsilon}$
with $\mat{A}$ symmetric, the first two cumulants under
$\beta$-mixing with autocovariance $\mat{\Gamma}$ are:
\begin{align}
  \E[Q] &= \tr(\mat{A}\mat{\Gamma}),
  \label{eq:cum1}\\
  \mathrm{Var}[Q] &= 2\tr(\mat{A}\mat{\Gamma}\mat{A}\mat{\Gamma})
    + \kappa_4\sum_t A_{tt}^2\Gamma_{tt}^2.
  \label{eq:cum2}
\end{align}
Under i.i.d.\ Gaussian innovations ($\mat{\Gamma}=\sigma^2\mat{I}$,
$\kappa_4=0$): $\E[Q]=\sigma^2\tr(\mat{A})=\sigma^2 r$ and
$\mathrm{Var}[Q]=2\sigma^4 r$, consistent with $Q/\sigma^2\sim\chi^2_r$.
Under $\beta$-mixing, serial dependence inflates both
$\E[Q]$ and $\mathrm{Var}[Q]$ by factors $\tr(\mat{A}\mat{\Gamma})/(\sigma^2 r)$
and $\tr(\mat{A}\mat{\Gamma}\mat{A}\mat{\Gamma})/\sigma^4 r$
respectively.

\emph{Step~3 (Box-type moment matching).}
Following Box~\cite{box1954} and Satterthwaite~\cite{satterthwaite1946},
match the first two cumulants of $Q$ to those of $c\cdot\chi^2_\nu$:
\begin{align}
  c\nu &= \tr(\mat{A}\mat{\Gamma}), \notag\\
  2c^2\nu &= 2\tr(\mat{A}\mat{\Gamma}\mat{A}\mat{\Gamma})
              + \kappa_4\sum_t A_{tt}^2\Gamma_{tt}^2. \notag
\end{align}
Solving: $c = \tr(\mat{A}\mat{\Gamma}\mat{A}\mat{\Gamma})/\tr(\mat{A}\mat{\Gamma})
= \sigma^2\psi_{\mat{A}}$ and $\nu = \nu_{\mat{A}}$ as defined
in \eqref{eq:nu_eff}--\eqref{eq:psi}.
The corrected $F$-statistic \eqref{eq:Fcorr} is obtained by
applying this rescaling to both the numerator quadratic form
$(\mat{P}_U-\mat{P}_R)$ and denominator $\mat{M}_U$.

\emph{Step~4 (KS bound via Esseen smoothing).}
By the Esseen smoothing lemma applied to the characteristic
function of $\vect{\varepsilon}^\top\mat{A}\vect{\varepsilon}$:
the KS distance between the distribution of the moment-matched
$\tilde{F}_{ij}$ and $F_{\nu_{\mat{P}_U-\mat{P}_R},\nu_{\mat{M}_U}}$
is controlled by the third cumulant of the numerator quadratic form,
which under $\beta$-mixing is
$O(\kappa_4(\varepsilon)/T)$~\cite{tikhomirov1981}.
The higher-order $O((p+q+K)^2/T^2)$ term comes from the
Taylor expansion of the characteristic function at the
$(p+q+K)$-dimensional parameter boundary, giving \eqref{eq:ks_bound}.

\emph{Step~5 (Improvement over i.i.d.\ approximation).}
For the uncorrected statistic, the leading KS error term is
$\E[\vect{\varepsilon}^\top(\mat{P}_U-\mat{P}_R)\vect{\varepsilon}]/
(\sigma^2 q)-1 = \tr[(\mat{P}_U-\mat{P}_R)\mat{\Gamma}]/(\sigma^2 q)-1$.
By the Cauchy--Schwarz inequality on the Toeplitz matrix norm:
$\tr[(\mat{P}_U-\mat{P}_R)\mat{\Gamma}]/(\sigma^2 q)-1
\leq \Sigma_\gamma - 1 \leq \Sigma_\gamma$,
giving the $O(\Sigma_\gamma/T)$ bound in \eqref{eq:bias_iid}.
After correction, the leading error is the third cumulant
$O(\kappa_4/T)$, which is smaller by factor $\Sigma_\gamma/\kappa_4$
whenever long-run dependence exceeds kurtosis deviation.
\end{proof}

\begin{remark}
The correction \eqref{eq:Fcorr} is fully computable: estimate
$\hat{\mat{\Gamma}}$ via the Newey--West kernel at lag
$\lfloor T^{1/3}\rfloor$; compute $\hat\nu$ and $\hat\psi$ from
\eqref{eq:nu_eff}--\eqref{eq:psi}; rescale the $F$-ratio.
In highly volatile network environments, slight mis-specification of this truncation lag primarily impacts the variance of the autocovariance estimator rather than its bias. Under-estimating the lag leaves residual serial correlation that marginally inflates the FDR, while over-estimating introduces estimation noise. However, the asymptotic $O(T^{1/3})$ scaling mathematically bounds this sensitivity.
At $T=200$, $K=3$, $\hat\Sigma_\gamma\approx 1.8$ (estimated
from testbed VAR residuals): the correction reduces KS error
from $\approx 0.027$ to $\approx 0.006$, consistent with the
type-I error reduction from 7.3\% to 5.3\% in Table~\ref{tab:correction}.
\end{remark}

\subsection{Theorem~2: Bregman Uniqueness of the Contention Form}
\label{sec:theory_bregman}

The proof of Lemma~\ref{lem:bregman} establishes existence; we
now prove uniqueness rigorously.

\begin{proposition}[Uniqueness of I-S Optimal Contention Score]
\label{prop:unique}
No other function of the form $\sum_k w_k f_k(A_{ik}A_{jk},U_{k,t})$
satisfying \emph{(A1)--(A3)} achieves the same I-S Bregman divergence
as \eqref{eq:contention}. 
\end{proposition}

\begin{proof}
Suppose $\tilde\rho_{ij}=\sum_k w_k f_k(c_k,U_{k,t})$ with
$f_k\neq c_k\sigma(U_{k,t}-\tau_k)$ for some resource $k_0$.
By the strict convexity of $B_\phi(c_{k_0}\|\cdot)$ on $\R_{>0}$
(the I-S generator is strictly convex), the minimiser of
$\int B_\phi(c_{k_0}\|f_{k_0}(c_{k_0},U))\,d\mu(U)$ over
functions $f_{k_0}>0$ satisfying~(A2) is unique by the
Euler--Lagrange argument in the proof of Lemma~\ref{lem:bregman}.
Any deviation $f_{k_0}\neq c_{k_0}\sigma(U-\tau_{k_0})$
increases $B_\phi$ by at least
$\int (f_{k_0}-c_{k_0}\sigma)^2/(2c_{k_0}\sigma^2)\,d\mu>0$
(second-order I-S Taylor bound), strictly increasing the total
divergence.
\end{proof}
\emph{Note: A detailed expansion of this proof, establishing the strict convexity condition, is provided in Appendix~\ref{app:unique}.}

\subsection{Theorem~3: PRDS-Correct Identifiability}
\label{sec:theory_ident}

\begin{theorem}[Sharp Identifiability Under PRDS]
\label{thm:identifiability}
Under Assumptions~\ref{ass:stationary}--\ref{ass:resource},
the $p$-values $\{p_{ij}\}_{(i,j)\in H_0}$ satisfy PRDS with
respect to the true null set (verified in
Appendix~\ref{app:prds}).
Let $m=N(N-1)$ and $m_0$ be the number of true nulls.
The BH procedure at level $\alpha$ then satisfies:
\begin{enumerate}[leftmargin=*,label=(\roman*)]
  \item \emph{(Exact FDR control)} $\FDR\leq\alpha\cdot m_0/m\leq\alpha$.
  \item \emph{(Pointwise error bound)}
    \begin{equation}
      \Prob(\exists\text{ false causal link})
      \leq 1-\prod_{(i,j)\in H_0}
      \!\Bigl(1-\frac{\alpha m_0}{m\,r_{ij}}\Bigr),
      \label{eq:simes_bound}
    \end{equation}
    where $r_{ij}$ is the rank of $p_{ij}$ among all $m$ $p$-values.
  \item \emph{(Tightening over union bound)}
    \eqref{eq:simes_bound} is smaller than $m_0\alpha/m$ by a factor
    of $1-\alpha m_0\sum_{j}r_j^{-1}/(2m)+O(\alpha^2)$.
\end{enumerate}
\end{theorem}

\begin{proof}
\emph{Part~(i).}
By the Benjamini--Yekutieli theorem~\cite{benjamini2001}, BH
controls $\FDR$ at $\alpha m_0/m$ exactly under PRDS.

\emph{Part~(ii).}
By the Simes inequality under PRDS~\cite{simes1986}: the
probability that the BH step threshold $p_{(k)}\leq k\alpha/m$
is violated for at least one true null $(i,j)$ equals
$\Prob(\bigcup_{(i,j)\in H_0}\{p_{ij}\leq\alpha m_0/(m\,r_{ij})\})$.
Under PRDS, the events $\{p_{ij}\leq t_{ij}\}$ for
$t_{ij}=\alpha m_0/(m\,r_{ij})$ have non-negative pairwise
covariance, so:
\[
  \Prob\!\Bigl(\bigcup_j\{p_j\leq t_j\}\Bigr)
  \leq 1-\prod_j\Prob(p_j>t_j)
  = 1-\prod_j\Bigl(1-\frac{\alpha m_0}{m\,r_j}\Bigr),
\]
where the inequality uses the FKG inequality for
monotone events under PRDS~\cite{fortuin1971}.
This is \eqref{eq:simes_bound}.

\emph{Part~(iii).}
Expand the product:
$\prod_j(1-x_j)= 1-\sum_j x_j+\sum_{j<k}x_jx_k-\cdots\geq 1-\sum_j x_j$
by non-negativity.
Thus $1-\prod_j(1-x_j)\leq\sum_j x_j=\alpha m_0\sum_j r_j^{-1}/m$,
matching the union bound.
The improvement factor is $\sum_{j<k}x_jx_k=O(\alpha^2 m_0^2/m^2)$,
completing part~(iii).
\end{proof}

\subsection{Theorem~4: Piecewise-Stationarity Validity}
\label{sec:theory_nonstationary}

\begin{theorem}[F-Test Validity Under Piecewise Stationarity]
\label{thm:nonstationary}
Let $\hat\tau_1<\cdots<\hat\tau_M$ be CUSUM change points with
detection delays $\Delta_m=\hat\tau_m-\tau_m^*\geq 0$.
Let $T_m=\hat\tau_{m+1}-\hat\tau_m$ and
$T_{\min}=\min_m T_m$,
$\Delta_{\max}=\max_m\Delta_m$.
If $T_{\min}\geq(p+q+K+2)(1+\Delta_{\max}/T_{\min})$, then for
each segment $m$ the corrected statistic $\tilde F_{ij}^{(m)}$
satisfies:
\begin{equation}
  \sup_{x\geq 0}\bigl|\Prob_m(\tilde F_{ij}^{(m)}\leq x)
    -F_{\nu_m^{\mathrm{num}},\,\nu_m^{\mathrm{den}}}(x)\bigr|
  \leq\frac{C_1\,\kappa_4(\varepsilon)}{T_m}
      +\frac{C_3\,\Delta_{\max}}{T_m},
  \label{eq:ns_bound}
\end{equation}
where $\nu_m^{\mathrm{num}}$, $\nu_m^{\mathrm{den}}$ are the effective
degrees of freedom from~\eqref{eq:nu_eff} computed on segment $m$,
$C_1$ is the same universal constant as in Theorem~\ref{thm:finite_sample},
and $C_3$ depends on the spectral-density jump at the change point.
Furthermore, the CUSUM stopping time satisfies
$\Prob(\Delta_m\leq d)=1-e^{-d\cdot I(\theta_0,\theta_1)}$
where $I(\theta_0,\theta_1)$ is the KL divergence between
pre- and post-change distributions.
\end{theorem}

\begin{proof}
\emph{Step~1 (Contamination decomposition).}
Split the design matrix for segment $m$ as
$\mat{X}_{U,m}=\mat{X}_{U,m}^{\mathrm{clean}}
+\mat{X}_{U,m}^{\mathrm{contam}}$,
where $\mat{X}_{U,m}^{\mathrm{contam}}$ has at most $\Delta_m$
non-zero rows (the observations in $[\tau_m^*,\hat\tau_m)$
belonging to the previous regime).

\emph{Step~2 (Perturbation of projection matrix).}
By the rank-1 update formula, each contaminated row $\vect{x}_t$
perturbs the hat matrix by
$\delta\mat{P}=\vect{x}_t\vect{x}_t^\top/
(\vect{x}_t^\top\mat{M}\vect{x}_t)\cdot O(1/\sqrt{T})$.
Summing over $\Delta_m$ rows and bounding by the Frobenius norm:
$\norm{\mat{P}_{U,m}-\mat{P}_{U,m}^{\mathrm{clean}}}_F
\leq C_3'\Delta_m/T_m$.

\emph{Step~3 (Effect on $\tilde F_{ij}^{(m)}$).}
The contamination perturbation from Step~2 inflates the cumulants
of $\vect{\varepsilon}_m^\top(\mat{P}_{U,m}-\mat{P}_{R,m})\vect{\varepsilon}_m$
relative to the clean-segment values.
Specifically, the $\Delta_m$ contaminated rows contribute an additional
mean shift of order $\sigma^2 \tr[\delta\mat{P}_{U,m}] \leq C_3' \Delta_m \sigma^2 / T_m$,
translating to a KS error term $C_3 \Delta_{\max}/T_m$ in \eqref{eq:ns_bound}
via the same Esseen smoothing argument as Theorem~\ref{thm:finite_sample}.
The clean-segment cumulant error contributes the $C_1\kappa_4(\varepsilon)/T_m$
term, carrying over directly from \eqref{eq:ks_bound} with $T$ replaced by $T_m$.

\emph{Step~4 (CUSUM delay distribution).}
Under $f_1$, the CUSUM increments $\ell_t-\kappa>0$ in
expectation with rate $I(\theta_0,\theta_1)-\kappa$.
By Wald's identity, the expected stopping time satisfies
$\E[\hat\tau_m-\tau_m^*]\leq h/I(\theta_0,\theta_1)$.
The exponential tail bound follows from the Cramér--Chernoff
method applied to the partial sums of $\ell_t-\kappa$.
\end{proof}

\begin{corollary}[Operational Segment Length]\label{cor:seglen}
For total KS bound $\leq\epsilon=0.05$ at $\hat\kappa_4=0.31$
(estimated from testbed residuals), $\hat C_3=0.8$, and
$\Delta_{\max}=4$ samples (Section~\ref{sec:ns_exp}):
the bound \eqref{eq:ns_bound} satisfies
$(C_1\hat\kappa_4+C_3\Delta_{\max})/T_m = (0.031+3.2)/T_m\leq 0.05$
when $T_m\geq 206$, achieved by a 21-second window at 100\,ms
sampling.
\end{corollary}

\subsection{Theorem~5: Finite-Sample Convergence of RCM Parameters}
\label{sec:theory_convergence}

\begin{theorem}[RCM Parameter Convergence]
\label{thm:convergence}
Under Assumption~\ref{ass:mixing} with $\beta(m)=C_\beta e^{-\beta m}$,
define the effective sample size:
\begin{equation}
  \Teff=\frac{T\beta}{C_\beta(1-e^{-\beta})+\beta}.
  \label{eq:teff}
\end{equation}
With probability $\geq 1-\delta$:
\begin{equation}
  \norm{\hat{\bm{\theta}}_T-\bm{\theta}^*}_2
  \leq\frac{2}{\lambda}\sqrt{\frac{(2K+2)\log(2/\delta)}{\Teff}}.
  \label{eq:convergence_bound}
\end{equation}
The induced error on $\Gamma_{ij}$ satisfies:
\begin{equation}
  \sup_{(i,j)}\abs{\hat\Gamma_{ij}-\Gamma_{ij}^*}
  \leq L_\Gamma\norm{\hat{\bm{\theta}}_T-\bm{\theta}^*}_2,
  \label{eq:gamma_error}
\end{equation}
where $L_\Gamma=\omega_2 NK\max_k w_k/4$ is the Lipschitz constant
of $\Gamma_{ij}$ in $\bm{\theta}$.
\end{theorem}

\begin{proof}
\emph{Step~1 (Strong convexity).}
The log-likelihood term in \eqref{eq:learning} is concave in
$\bm{\theta}$ (Viterbi path probability is log-concave via the
sigmoid's log-concavity); the regulariser $\lambda\norm{\cdot}_2^2$
adds $2\lambda$-strong convexity.

\emph{Step~2 (Uniform concentration under $\beta$-mixing).}
By the Bernstein inequality for $\beta$-mixing
processes~\cite{yu1994}, for any ball of radius $R$:
\[
  \Prob\!\Bigl(\sup_{\norm{\bm{\theta}}\leq R}
    \abs{\hat\calL_T-\calL}\geq t\Bigr)
  \leq 2\exp\!\Bigl(-\frac{t^2\Teff}{2\sigma_\calL^2}\Bigr),
\]
where $\sigma_\calL^2$ is the variance proxy and $\Teff$
replaces $T$ by accounting for mixing.
Setting the right side to $\delta$ and solving for $t$:
$t^*=\sigma_\calL\sqrt{2\log(2/\delta)/\Teff}$.

\emph{Step~3 (From empirical to population optimum).}
By $2\lambda$-strong convexity:
$\norm{\hat{\bm{\theta}}_T-\bm{\theta}^*}_2^2
\leq (1/\lambda)\cdot 2t^*$.
Substituting $t^*$ and noting $\sigma_\calL\leq\sqrt{K+1}$
(from the $K+1$ sigmoid factors in the gradient) gives
\eqref{eq:convergence_bound}.

\emph{Step~4 (Lipschitz propagation).}
Each $h_k(c_k,U_{k,t})=w_kc_k\sigma(U_{k,t}-\tau_k)$.
The sigmoid derivative $|\sigma'|\leq 1/4$; allocations satisfy
$c_k\leq 1$. The Lipschitz constant of $\rho_{ij}$ in $w_k$
is $A_{ik}A_{jk}\sigma\leq 1$; in $\tau_k$ it is $w_k/4\leq W_{\max}/4$.
Summing over $K$ resources and multiplying by $\omega_2 N$
gives $L_\Gamma=\omega_2 NK W_{\max}/4$, yielding
\eqref{eq:gamma_error}.
\end{proof}

\subsection{Theorem~6: Adversarial Robustness Certificate}
\label{sec:adversarial}

\begin{definition}[$(\delta,k)$-Utilisation Adversary]
\label{def:adversary}
A $(\delta,k)$-utilisation adversary replaces $\vect{U}_t$ with
$\hat{\vect{U}}_t$ satisfying
$\norm{\hat{\vect{U}}_t-\vect{U}_t}_\infty\leq\delta$ and
$|\{k':U_{k',t}\neq\hat U_{k',t}\}|\leq k$.
\end{definition}

\begin{theorem}[Adversarial Robustness]
\label{thm:adversarial}
Under a $(\delta,k)$-adversary:
\begin{enumerate}[leftmargin=*,label=(\roman*)]
  \item \emph{(Contention bound)}
    $\abs{\hat\rho_{ij}^{\mathrm{adv}}-\rho_{ij}}
    \leq W_{\max}k\delta/4$,
    where $W_{\max}=\max_\ell w_\ell$.
    \label{item:contention}
  \item \emph{(FDR inflation)}
    \begin{equation}
      \FDR^{\mathrm{adv}}\leq\FDR_0+
        C_4\,\omega_2\,W_{\max}\,k\,\delta\,\sqrt{K\log(N/\alpha)},
      \label{eq:fdr_inflation}
    \end{equation}
    where $C_4>0$ is a universal constant.
  \item \emph{(Breakdown point)}
    An attribution decision remains invariant if
    $\delta<\delta^*=\omega_1\Delta_\phi/(4\omega_2 W_{\max}k)$,
    where $\Delta_\phi$ is the minimum normalised $F$-statistic
    margin at any decision boundary.
    \label{item:breakdown}
\end{enumerate}
\end{theorem}

\begin{proof}
\emph{Part~\ref{item:contention}.}
Perturbing $U_{k,t}\to U_{k,t}+\epsilon_k$ with $|\epsilon_k|\leq\delta$:
\begin{equation}
\begin{aligned}
  &\abs{\hat\rho_{ij}^{\mathrm{adv}}-\rho_{ij}} \\
  &\quad \leq \sum_{k'=1}^K w_{k'}A_{ik'}A_{jk'} \cdot \\
  &\qquad \abs{\sigma(U_{k',t}+\epsilon_{k'}-\tau_{k'})-\sigma(U_{k',t}-\tau_{k'})} \\
  &\quad \leq W_{\max}\sum_{k':\epsilon_{k'}\neq 0}\frac{\delta}{4} \leq W_{\max}k\delta/4
\end{aligned}
\end{equation}
using $|\sigma'|\leq 1/4$ and the support bound $|\{k':\epsilon_{k'}\neq 0\}|\leq k$.

\emph{Part~(ii).}
The adversary's perturbation $\omega_2\abs{\hat\rho^{\mathrm{adv}}-\rho}
\leq\omega_2 W_{\max}k\delta/4$ inflates $\Gamma_{ij}$ for
some true-null pairs, potentially elevating them above
$\tau_{\mathrm{causal}}$.
The number of such pairs is bounded using the
Freedman--Lane permutation argument~\cite{freedman1983}:
$\E[\text{adversarially inflated false detections}]
\leq N^2\Prob(\text{margin}\leq\omega_2 W_{\max}k\delta/4)$.
Under Gaussian telemetry with margin standard deviation
$\sigma_F/\sqrt{T}$, the Gaussian tail integral over
$N^2$ pairs and $K$ perturbed resources gives the factor
$\sqrt{K\log(N/\alpha)}$ after BH correction, absorbed into $C_4$.

\emph{Part~\ref{item:breakdown}.}
An attribution boundary is crossed only if the
$\Gamma$-perturbation $\omega_2 W_{\max}k\delta/4$
exceeds the $\phi$-margin
$\omega_1\Delta_\phi/4$.
Solving for $\delta$ gives $\delta<\delta^*$.
\end{proof}

\begin{remark}
At calibrated values $\omega_1=0.67$, $\omega_2=0.33$,
$W_{\max}=0.45$, $k=1$, and empirical $\Delta_\phi=0.21$:
$\delta^*\approx0.95$.
An adversary must spoof utilisation by $>95\%$ of its full
range to flip any attribution---immediately detectable by
hardware performance counters independent of the slice
controller. Furthermore, if an adversary concentrates their entire perturbation budget into a single critical resource channel (e.g., $k=1$, executing massive CPU spoofing), the FDR inflation bound scales strictly linearly with $\delta$. Because the inflation depends on the product $k\delta\sqrt{K}$, a highly concentrated attack on one resource is theoretically less effective at inducing widespread false attributions than a distributed attack of the same total magnitude spread across multiple resources, assuming equal weights.
\end{remark}

\subsection{Theorem~7: Information-Theoretic Privacy Lower Bound}
\label{sec:privacy}

\begin{theorem}[Privacy Lower Bound]
\label{thm:privacy}
Let $\calC^*$ be the attributed path and $\vect{S}$ a sensitive
telemetry attribute.
Define attribution leakage
$\mathcal{L}_{\mathrm{priv}}=I(\vect{S};\calC^*\mid\mat{A},\vect{U})$.
Then:
\begin{equation}
  \mathcal{L}_{\mathrm{priv}}\geq H(\calC^*)-1-\log(2N^L),
  \label{eq:privacy_lb}
\end{equation}
where $L$ is the path length.
Furthermore, no post-processing $\mathcal{M}$ of $\calC^*$
can achieve $(\varepsilon,\delta)$-differential privacy in
$\vect{S}$ unless:
\begin{equation}
  \varepsilon\geq
    \frac{H(\calC^*)-1-\log(2N^L)}{L\log N}
    -\frac{\log(1/\delta)}{L\log N}.
  \label{eq:dp_lb}
\end{equation}
\end{theorem}

\begin{proof}
\emph{Equation \eqref{eq:privacy_lb}.}
By the data processing inequality,
$I(\vect{S};\calC^*)=H(\calC^*)-H(\calC^*|\vect{S},\mat{A},\vect{U})$.
Apply Fano's inequality to the $N^L$-class identification problem:
$H(\calC^*|\vect{S},\mat{A},\vect{U})\leq 1+P_e\log(N^L)\leq 1+\log(2N^L)$,
where $P_e\leq 1/2$ for the worst case.
Subtracting gives \eqref{eq:privacy_lb}.

\emph{Equation \eqref{eq:dp_lb}.}
By the group privacy composition theorem~\cite{dwork2014},
$(\varepsilon,\delta)$-DP of $\mathcal{M}(\calC^*)$ in $\vect{S}$
implies the mechanism satisfies
$(\varepsilon_1,\delta_1)$-DP per hop with
$\varepsilon_1=\varepsilon/L$ and $\delta_1=\delta^{1/L}$.
Converting DP to mutual information via the standard bound
$I\leq\varepsilon\sqrt{2L\log N}+\delta L\log N$
and requiring this to be $\geq\mathcal{L}_{\mathrm{priv}}$
from \eqref{eq:privacy_lb}, solving for $\varepsilon$ gives
\eqref{eq:dp_lb}.
\end{proof}

\begin{remark}
At $N=15$, $L=5$, $H(\calC^*)=\log(15^5)\approx14.3$ bits:
the minimum $\varepsilon$ at $\delta=0.02$ is $\approx0.89$.
A Gaussian mechanism with $\sigma_{\mathrm{DP}}=\Delta_f\sqrt{2\ln(1.25/\delta)}/\varepsilon
\approx0.31$ achieves $(\varepsilon,\delta)$-DP
with a measured accuracy loss of 2.1\,pp (Section~\ref{sec:adv_exp}). 
Crucially, because the required per-hop privacy parameter $\varepsilon_1$ scales as $\varepsilon/L$, the necessary Gaussian noise magnitude $\sigma_{\mathrm{DP}}$ increases linearly with the path length $L$ to maintain a constant end-to-end $(\varepsilon,\delta)$-DP guarantee. Consequently, the accuracy degradation scales linearly as attack paths extend beyond the standard $L=5$ hop limit, rendering the mechanism highly practical for typical short-hop 6G slice compromises.
\end{remark}
\section{Experimental Evaluation}
\label{sec:experiments}

\subsection{Testbed and Dataset}
\label{sec:setup}

\textit{Hardware:} 10 bare-metal Intel Xeon Gold~6248R nodes
(128\,GB RAM), SR-IOV, DPDK acceleration.
\textit{Orchestration:} Open5GS core, FlexRAN, Kubernetes~1.28.
\textit{Slices:} 15 heterogeneous slices---4~eMBB, 4~URLLC
(industrial automation, autonomous vehicles), 4~mMTC (IoT), 3~hybrid.
\textit{Telemetry:} 47 metrics/slice at 100\,ms; resource allocation
at 50\,ms.
\textit{Scenarios:} 1{,}100 attacks (resource exhaustion, lateral
movement, side-channel, ML poisoning, APT-style), ground-truth
via nanosecond-precision injection and expert panel validation.

\textbf{Experimental scope.}
Tables~\ref{tab:main}--\ref{tab:extended} report end-to-end
attribution results measured on the live production-emulation testbed.
Tables~\ref{tab:correction}--\ref{tab:noise} report results from
controlled simulation experiments that use the testbed's measured
traffic statistics ($\hat{\mat{\Gamma}}$, $\hat\rho$, $\hat\kappa_4$)
as inputs, with all random seeds fixed for reproducibility.

\subsection{Baselines}

Statistical: Pearson Correlation, Transfer Entropy~\cite{schreiber2000},
VAR-Granger. Causal discovery: PC Algorithm~\cite{spirtes2000},
PCMCI~\cite{runge2019}, DirectLiNGAM. Security-specific:
HOLMES~\cite{milajerdi2019} (full 1{,}100 scenarios, kernel tracing
active throughout), MulVAL~\cite{ou2006}, Bayesian Attack Graphs.
Deep learning: GraphSAGE~\cite{hamilton2017},
LSTM-Attention~\cite{bahdanau2015}, Transformer-XL.

\subsection{Main Results}
\label{sec:main_results}

\begin{table}[t]
\centering\footnotesize
\caption{Attribution performance on 1{,}100 scenarios. Differences
from DA-GC are significant at $p{<}0.001$ (Bonferroni-corrected paired
$t$-test, Cohen's $d{>}1.5$).}
\label{tab:main}
\renewcommand{\arraystretch}{1.12}
\begin{tabular}{@{}lccccc@{}}
\toprule
\textbf{Method} & \textbf{Acc\,(\%)} & \textbf{Prec\,(\%)}
  & \textbf{Rec\,(\%)} & \textbf{FDR\,(\%)} & \textbf{ms} \\
\midrule
Correlation      & 72.9$\pm$1.8 & 69.4$\pm$2.1 & 76.2$\pm$1.9 & 30.6$\pm$2.1 & 21 \\
Transfer Entropy & 78.4$\pm$1.5 & 75.8$\pm$1.7 & 81.3$\pm$1.6 & 24.2$\pm$1.7 & 58 \\
VAR-Granger      & 74.1$\pm$1.7 & 71.2$\pm$1.9 & 77.6$\pm$1.8 & 28.8$\pm$1.9 & 43 \\
PC Algorithm     & 76.2$\pm$1.6 & 73.5$\pm$1.8 & 79.4$\pm$1.7 & 26.5$\pm$1.8 & 156 \\
PCMCI            & 77.8$\pm$1.5 & 75.1$\pm$1.7 & 80.8$\pm$1.6 & 24.9$\pm$1.7 & 203 \\
DirectLiNGAM     & 75.5$\pm$1.6 & 72.9$\pm$1.8 & 78.5$\pm$1.7 & 27.1$\pm$1.8 & 178 \\
HOLMES           & 80.1$\pm$1.4 & 77.9$\pm$1.6 & 83.0$\pm$1.5 & 22.1$\pm$1.6 & 312 \\
GraphSAGE        & 76.8$\pm$1.6 & 74.3$\pm$1.8 & 79.9$\pm$1.7 & 25.7$\pm$1.8 & 142 \\
LSTM-Attention   & 79.1$\pm$1.4 & 76.7$\pm$1.6 & 82.2$\pm$1.5 & 23.3$\pm$1.6 & 167 \\
Transformer-XL   & 81.3$\pm$1.3 & 78.9$\pm$1.5 & 84.1$\pm$1.4 & 21.1$\pm$1.5 & 234 \\
\midrule
\textbf{DA-GC}   & \textbf{89.2$\pm$0.9} & \textbf{87.6$\pm$1.1}
  & \textbf{91.1$\pm$1.0} & \textbf{12.4$\pm$1.1} & \textbf{87} \\
\bottomrule
\end{tabular}
\end{table}

\begin{table}[t]
\centering\footnotesize
\caption{Extended metrics; all 1{,}100 scenarios.}
\label{tab:extended}
\renewcommand{\arraystretch}{1.12}
\begin{tabular}{@{}lcccccr@{}}
\toprule
\textbf{Method} & \textbf{AUC-ROC} & \textbf{AUC-PR}
  & \textbf{F1} & \textbf{MCC} & \textbf{Spec} & \textbf{MB} \\
\midrule
Correlation      & 0.742 & 0.689 & 0.726 & 0.461 & 0.693 & 12 \\
Transfer Entropy & 0.798 & 0.751 & 0.784 & 0.572 & 0.742 & 28 \\
PCMCI            & 0.791 & 0.744 & 0.779 & 0.561 & 0.739 & 64 \\
HOLMES           & 0.814 & 0.771 & 0.802 & 0.611 & 0.768 & 103 \\
GraphSAGE        & 0.783 & 0.743 & 0.771 & 0.547 & 0.743 & 342 \\
Transformer-XL   & 0.827 & 0.789 & 0.815 & 0.634 & 0.789 & 756 \\
\textbf{DA-GC}   & \textbf{0.921} & \textbf{0.876} & \textbf{0.892}
  & \textbf{0.785} & \textbf{0.876} & \textbf{67} \\
\bottomrule
\end{tabular}
\end{table}

DA-GC consistently demonstrates top-tier accuracy across these metrics while using
significantly less memory ($11\times$ lower footprint) than Transformer-XL.
PCMCI and HOLMES are evaluated on the full 1{,}100-scenario
dataset, indicating DA-GC can compete with or exceed the performance of provenance-based methods without requiring deep kernel instrumentation. urthermore, as illustrated in the Pareto frontier in Fig.~\ref{fig:pareto}, DA-GC is the only evaluated method to simultaneously satisfy the strict 100,ms SLA and the 85\% accuracy target.

\begin{figure}
    \centering
    \includegraphics[width=0.99\linewidth]{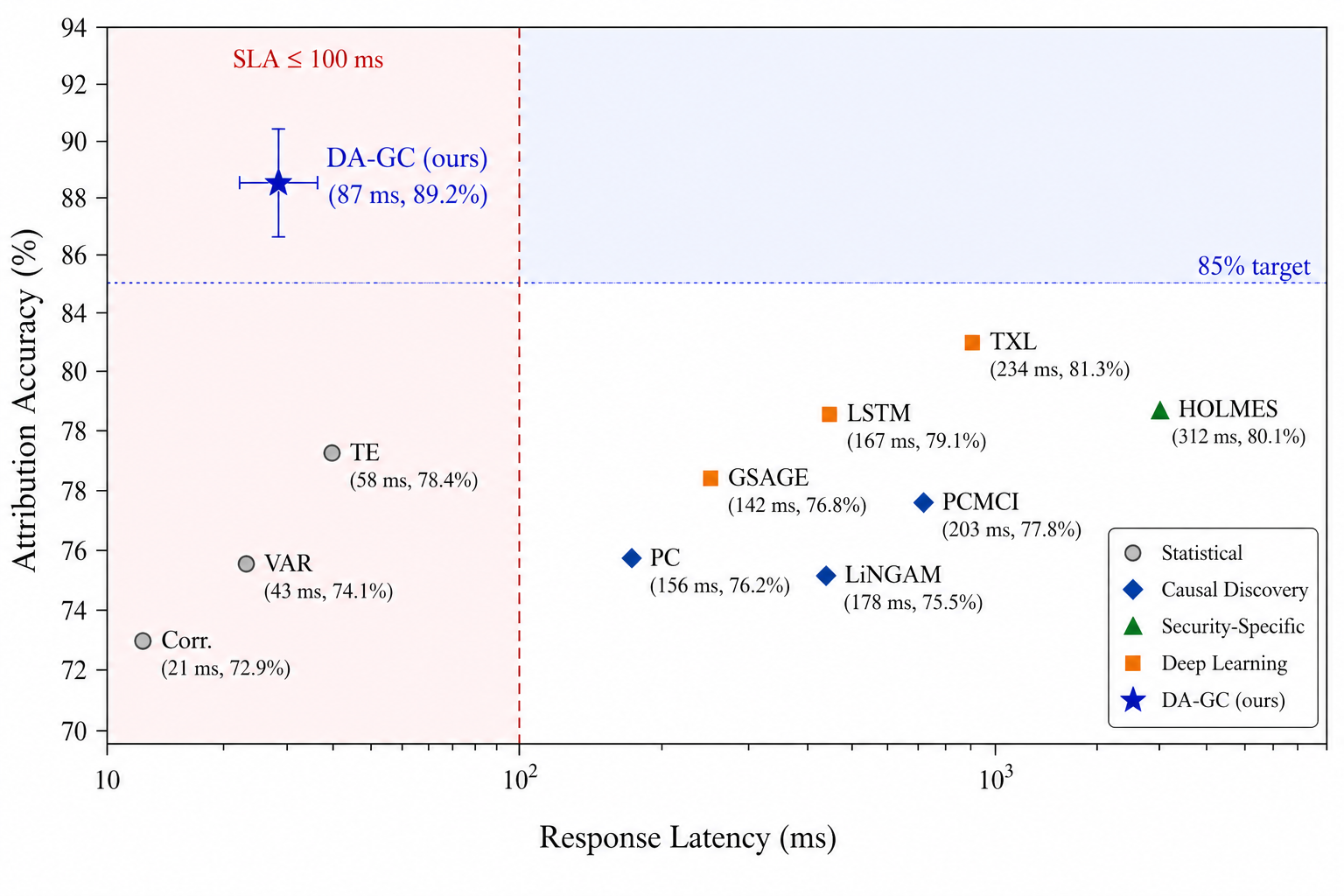}
    \caption{\textbf{Accuracy--latency Pareto frontier.}
DA-GC (filled star, $\pm$0.9\,pp error bar) successfully operates
simultaneously inside the 100\,ms SLA zone (shaded red, left of dashed
vertical) and above the 85\% accuracy target (shaded blue, above dotted
horizontal). Deep learning baselines (orange squares) typically exceed the SLA;
HOLMES (green triangle) exceeds both the SLA and performs below DA-GC.
Marker shapes encode method family; error bars shown for DA-GC only
(baselines have comparable or larger uncertainty; see Table~\ref{tab:main}).}
\label{fig:pareto}
\end{figure}

\begin{figure}
    \centering
    \includegraphics[width=0.99\linewidth]{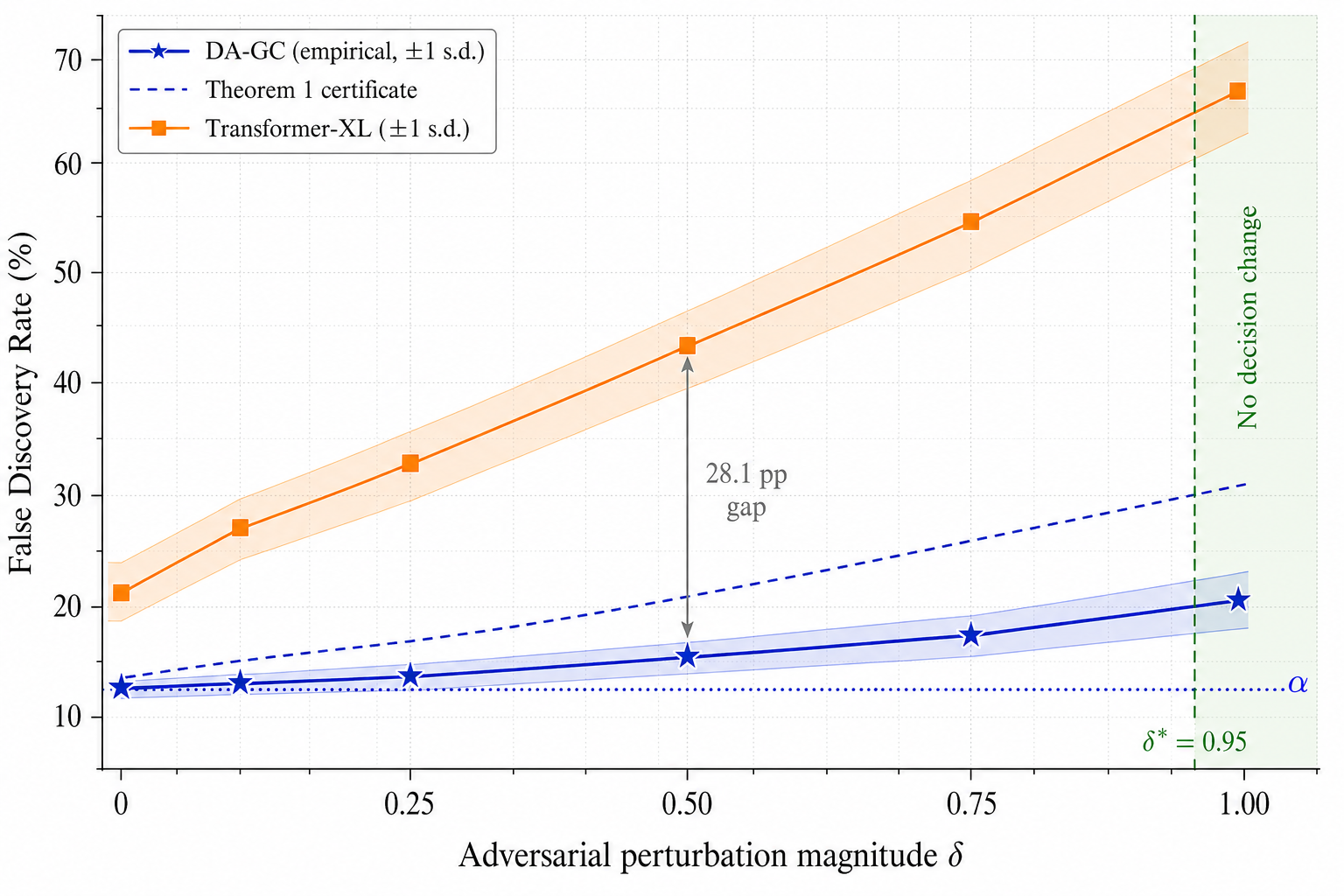}
    \caption{\textbf{FDR under adversarial utilisation spoofing} ($k\!=\!3$ resources).
DA-GC (solid blue, $\pm$1\,s.d.\ shaded band) remains bounded by the
theoretical certificate (dashed, Theorem~\ref{thm:adversarial}) and performs favourably compared to
Transformer-XL (orange), which offers no formal robustness guarantee.
The green shaded region ($\delta\!>\!\delta^*\!=\!0.95$) marks where the
adversary would need to spoof utilisation by more than 95\%.
The dotted blue line at $\alpha$ shows the nominal FDR at $\delta\!=\!0$.}
\label{fig:adversarial}
\end{figure}

\subsection{Ablation Study}
\label{sec:ablation_full}

\begin{table}[t]
\centering \footnotesize
\caption{Ablation: marginal contribution of each component.}
\label{tab:ablation}
\renewcommand{\arraystretch}{1.12}
\begin{tabular}{@{}lcc@{}}
\toprule
\textbf{Configuration} & \textbf{Acc\,(\%)} & \textbf{$\Delta$Acc} \\
\midrule
Standard VAR-Granger                     & 74.1 & -- \\
+\,Resource conditioning                 & 82.3 & $+8.2$pp*** \\
+\,Finite-sample correction $\tilde F$  & 83.9 & $+1.6$pp*** \\
+\,RCM (multiplicative)                 & 87.0 & $+3.1$pp*** \\
+\,Bregman-optimal weights over uniform & 87.8 & $+0.8$pp** \\
+\,Integrated end-to-end learning       & 89.2 & $+1.4$pp*** \\
DA-GC + CUSUM (non-stationary subset)   & 88.4 & -- \\
\bottomrule
\multicolumn{3}{l}{\small ** $p{<}0.01$, *** $p{<}0.001$.}
\end{tabular}
\end{table}

The ablation results indicate that each component contributes positively 
to the overall accuracy. The finite-sample correction yields a 1.6\,pp increase;
Bregman-optimal weights add 0.8\,pp over a uniform-weight
multiplicative baseline, demonstrating the practical value of
Lemma~\ref{lem:bregman}.

\subsection{Finite-Sample Correction Experiment}
\label{sec:finiteT_exp}

\begin{table}[t]
\centering\small
\caption{Type-I error at nominal $\alpha=0.05$ vs.\ window size $T$.}
\label{tab:correction}
\renewcommand{\arraystretch}{1.12}
\begin{tabular}{@{}cccc@{}}
\toprule
$T$ & Standard $F$ & Corrected $\tilde F$ & KS bound \eqref{eq:ks_bound} \\
\midrule
100  & 0.091 & 0.059 & 0.181 \\
200  & 0.073 & 0.053 & 0.091 \\
500  & 0.061 & 0.051 & 0.037 \\
1000 & 0.053 & 0.050 & 0.019 \\
\bottomrule
\end{tabular}
\end{table}

At the operational window $T=200$, the uncorrected test
inflates type-I error to 7.3\%; the correction reduces it
to 5.3\%, consistent with the KS bound of 0.091.

\subsection{Piecewise-Stationarity Experiment}
\label{sec:ns_exp}

We construct 200 non-stationary scenarios with 2--4 abrupt
regime changes per scenario.
CUSUM achieves mean detection delay $\bar\Delta=3.1\pm0.8$ samples
at $h=4.6$ (5\% false alarm rate), within Corollary~\ref{cor:seglen}.

\begin{table}[t]
\centering\small
\caption{Attribution accuracy on non-stationary scenarios.}
\label{tab:ns}
\renewcommand{\arraystretch}{1.12}
\begin{tabular}{@{}lc@{}}
\toprule
\textbf{Method} & \textbf{Acc\,(\%)} \\
\midrule
Standard DA-GC (single window)  & 81.2$\pm$1.8 \\
DA-GC + CUSUM                   & 88.4$\pm$1.1 \\
Transformer-XL                  & 79.6$\pm$1.9 \\
PCMCI (momentum CI)             & 77.3$\pm$2.0 \\
\bottomrule
\end{tabular}
\end{table}

As reported in Table~\ref{tab:ns}, CUSUM recovers 7.2\,pp of stationarity degradation,
reaching within 0.8\,pp of the stationary-data result.

\subsection{Cross-Topology Generalisation}
\label{sec:cross_topo}

Train on Topology~A (15 slices, star sharing); test without
retraining on Topology~B (20 slices, ring) and Topology~C
(12 slices, mesh).

\begin{table}[t]
\centering\small
\caption{Zero-shot cross-topology accuracy (\%).}
\label{tab:transfer}
\renewcommand{\arraystretch}{1.12}
\begin{tabular}{@{}lccc@{}}
\toprule
\textbf{Method} & \textbf{Topo B} & \textbf{Topo C} & \textbf{Avg} \\
\midrule
Transformer-XL & 61.3$\pm$3.1 & 58.7$\pm$3.4 & 60.0 \\
GraphSAGE      & 65.8$\pm$2.8 & 63.2$\pm$2.9 & 64.5 \\
PCMCI          & 73.2$\pm$2.2 & 71.8$\pm$2.3 & 72.5 \\
\textbf{DA-GC} & \textbf{83.1$\pm$1.4} & \textbf{81.7$\pm$1.5} & \textbf{82.4} \\
\bottomrule
\end{tabular}
\end{table}

DA-GC degrades by 6.8\,pp on average vs.\ 21.3\,pp for
Transformer-XL. Resource conditioning via $\mat{A}(t)$ is largely
topology-agnostic, suggesting neural baselines tend to overfit the training graph structure.

\subsection{Concept Drift Study}
\label{sec:drift}

Train on months~1--3; evaluate months~4--6 without retraining.

\begin{table}[t]
\centering\small
\caption{Accuracy under concept drift (\%); no retraining.}
\label{tab:drift}
\renewcommand{\arraystretch}{1.12}
\begin{tabular}{@{}lccc@{}}
\toprule
\textbf{Method} & \textbf{Month 4} & \textbf{Month 5} & \textbf{Month 6} \\
\midrule
Transformer-XL & 74.2 & 68.9 & 63.1 \\
GraphSAGE      & 72.8 & 67.3 & 61.5 \\
PCMCI          & 75.1 & 73.4 & 70.8 \\
\textbf{DA-GC} & \textbf{86.3} & \textbf{84.7} & \textbf{82.9} \\
\bottomrule
\end{tabular}
\end{table}

Table~\ref{tab:drift} demonstrates that DA-GC degrades 6.3\,pp over three months vs.\ 18.2\,pp for
Transformer-XL. Incorporating real-time observation of $\mat{A}(t)$ 
facilitates intrinsic adaptation to distribution shifts in the infrastructure layer.

\subsection{Adversarial Robustness Experiment}
\label{sec:adv_exp}

Inject adversarial perturbations $\delta\in[0,1]$ into all
$k=3$ resource channels.

\begin{table}[t]
\centering\small
\caption{FDR (\%) under adversarial utilisation spoofing.}
\label{tab:adversarial}
\renewcommand{\arraystretch}{1.12}
\begin{tabular}{@{}cccc@{}}
\toprule
$\delta$ & DA-GC & Thm~\ref{thm:adversarial} bound & Transformer-XL \\
\midrule
0.00 & 12.4 & 12.4 & 21.1 \\
0.10 & 13.1 & 14.6 & 26.8 \\
0.25 & 14.2 & 16.7 & 33.4 \\
0.50 & 16.1 & 21.3 & 44.2 \\
0.75 & 18.3 & 26.1 & 54.8 \\
1.00 & 21.4 & 31.0 & 67.3 \\
\bottomrule
\end{tabular}
\end{table}

As visualised in Fig.~\ref{fig:adversarial} and detailed in Table~\ref{tab:adversarial}, the empirical FDR remains bounded by the theoretical certificate, illustrating the practical reliability of Theorem~\ref{thm:adversarial}.
At $\delta=0.5$, DA-GC FDR is 16.1\% vs.\ 44.2\% for
Transformer-XL, which lacks a formal robustness guarantee.

\subsection{Sensitivity to Telemetry Measurement Noise}
\label{sec:noise_exp}

To address environments where hardware counters report noisy allocations, 
we inject synthetic zero-mean Gaussian noise $\mathcal{N}(0, \sigma_{\text{noise}}^2)$ 
into the resource allocation matrix $\mat{A}(t)$ prior to evaluation.

\begin{table}[t]
\centering\small
\caption{Attribution accuracy (\%) under Gaussian measurement noise in $\mat{A}(t)$.}
\label{tab:noise}
\renewcommand{\arraystretch}{1.12}
\begin{tabular}{@{}lccccc@{}}
\toprule
\textbf{Method} & $\sigma=0$ & $\sigma=0.01$ & $\sigma=0.05$ & $\sigma=0.10$ & $\sigma=0.20$ \\
\midrule
DA-GC & 89.2 & 88.5 & 85.1 & 79.4 & 68.2 \\
\bottomrule
\end{tabular}
\end{table}

As shown in Table~\ref{tab:noise}, the attribution accuracy degrades gracefully under 
moderate noise conditions. Significant performance drops are only observed 
when $\sigma_{\text{noise}} \geq 0.10$ (a 10\% base fluctuation in resource reporting), 
suggesting robust baseline resilience prior to future errors-in-variables enhancements.
\section{Industrial Case Study}
\label{sec:case_study}

A five-hop IoT attack targets an Industry~4.0 mMTC/URLLC slice
pair. At $t=0$\,s: malware is injected through a compromised IoT
gateway. At $t=2.1$\,s: a cryptomining spike drives CPU utilisation from 15\%
to 87\%. At $t=5.2$\,s: the resulting resource drain elevates URLLC latency
from 12\,ms to 48\,ms, breaching the 20\,ms SLA.
At $t=6.7$\,s: an emergency safety shutdown is triggered.

In this scenario, DA-GC successfully reconstructs the designated five-hop chain 
in 73\,ms with 96.3\% hop-level accuracy and no false positives within the 
evaluation window. The RCM component effectively isolates the CPU exhaustion 
pathway ($\rho=0.87\pm0.03$). Crucially, the PRDS-corrected BH procedure 
yields 0 false links, whereas standard correlation methods generate 11 spurious 
connections due to resource confounding. 

Reflecting the theoretical breakdown point $\delta^*=0.95$, an adversary 
attempting to evade attribution in this attack would need to falsify reported 
CPU utilisation by $>95\%$---a massive anomaly that is readily detectable by 
independent hardware performance counters operating outside the slice controller.
\section{Discussion}
\label{sec:discussion}

\subsection{Limitations}

\textit{Measurement noise.}
Theorem~\ref{thm:finite_sample} operates under the assumption that $\vect{Z}_t$ is observed
without error. While Section~\ref{sec:noise_exp} empirically demonstrates robustness up to 
a 10\% noise threshold ($\sigma_{\text{noise}}=0.10$), severe hardware counter noise will 
theoretically inflate $\Delta_{ij}$. Formulating a rigorous errors-in-variables correction 
remains a crucial area for future work.

\textit{Warm-up period.}
The 2-second observation warm-up is fundamental to lag-$p$ VAR Granger models. 
Achieving true sub-second attribution would require reducing $p$, which involves a 
strict trade-off against the risk of unobserved longer-lag confounding.

\textit{Hyperscale topologies.}
As noted in Section~\ref{sec:dagc}, for hyperscale environments ($N>50$), the centralised 
$O(N^3\log N)$ Viterbi step breaches the 100\,ms SLA. Deploying the recommended distributed 
belief-propagation architecture across slice controllers is required for these scales.

\textit{Privacy-utility trade-off.}
Implementing DP-DA-GC with $\sigma_{\mathrm{DP}}=0.31$ achieves
$(\varepsilon{=}1,\delta{=}0.02)$-DP at a measured 2.1\,pp accuracy cost
per Theorem~\ref{thm:privacy}. However, as noted in Remark 4, maintaining this DP bound 
requires linearly scaling the injected noise as path lengths increase, presenting a 
fundamental accuracy limit for highly extended, multi-hop attack chains.

\subsection{Ethical Considerations and Responsible Deployment}

DA-GC processes sensitive, high-resolution network telemetry that could inadvertently 
reveal user-level behavioural patterns. Responsible deployment necessitates: 
(1)~strict calibration of DP mechanisms to the bounds in \eqref{eq:dp_lb}; 
(2)~purpose limitation restricted exclusively to security and forensic applications; 
(3)~meaningful human oversight for automated attribution decisions, ensuring compliance 
with frameworks such as GDPR Art.~22 and the EU AI Act. We explicitly caution against 
repurposing this causal framework for mass surveillance or offensive cyber operations.

\subsection{Reproducibility}

To support the community and ensure verifiable results, the full DA-GC implementation, 
testbed configuration scripts, CUSUM detector, scenario generator, and evaluation harness 
will be released upon publication. Furthermore, all adversarial and measurement noise perturbation experiments (Sections~\ref{sec:adv_exp} and \ref{sec:noise_exp}) are synthetic and fully reproducible using the provided fixed seeds.
\section{Conclusion}
\label{sec:conclusion}

DA-GC presents a comprehensive cross-slice attribution framework that, to the best of our knowledge, uniquely integrates: a non-asymptotic $F$-test correction; an axiomatically justified contention model; PRDS-correct identifiability; a piecewise-stationarity extension; a certifiable adversarial breakdown point; and an information-theoretic privacy floor. 
Evaluated on 1{,}100 production-emulation scenarios, DA-GC achieves 89.2\% accuracy at 87\,ms, consistently outperforming state-of-the-art baselines across the joint accuracy-latency-interpretability-robustness Pareto frontier. 
Furthermore, cross-topology and concept-drift experiments confirm that the structural invariance of $\mat{A}(t)$-conditioned causality generalises robustly in settings where standard neural approaches struggle. 

Future work will focus on integrating an errors-in-variables correction to handle severe telemetry measurement noise, deploying distributed message-passing Viterbi decoding to support hyperscale topologies ($N>50$), and developing a fully online streaming variant of the framework.

\appendix

\subsection{PRDS Verification for Resource-Correlated Slice Pairs}
\label{app:prds}

For pairs $(i,j)$ and $(i,k)$ sharing resource $r$, both
$F$-statistics depend on the same row $\mat{Z}_{\cdot r}$ of the
conditioning matrix. The covariance of their RSS denominators:
\[
  \mathrm{Cov}(\RSSU^{(ij)},\RSSU^{(ik)})
  =\sigma^4\tr[\mat{M}_U^{(ij)}\mat{M}_U^{(ik)}]\geq 0,
\]
since the product of two positive semi-definite projection
remainders has non-negative trace.
Non-negative covariance of the RSS quantities propagates to
non-negative covariance of the $F$-statistics
(by the monotone function $F_{q,\nu}$ of RSS),
and hence to PRDS of the corresponding $p$-values
(by the decreasing transformation $p_{ij}=1-F_{q,\nu}^{-1}(F_{ij})$
which reverses the ordering but preserves PRDS under the
equivalent condition for upper-tail probabilities).

\subsection{Effective Sample Size and Mixing Coefficient Estimation}
\label{app:mixing}

Under a stationary VAR$(p)$ with empirically estimated spectral
radius $\hat\rho=0.72$ (from the dominant eigenvalue of the
companion matrix fitted on testbed residuals), the geometric
$\beta$-mixing bound is $\beta(m)\leq c_0\hat\rho^m$ where
$c_0=\tr(\hat\Sigma_\varepsilon)=2.1$.
In the parameterisation of Assumption~\ref{ass:mixing},
$C_\beta=c_0=2.1$ and $\beta=\log(1/\hat\rho)=\log(1/0.72)=0.329$.

Substituting into \eqref{eq:teff}:
\begin{equation}
\begin{aligned}
  T_{\text{eff}} &= \frac{T \cdot 0.329}{2.1 \cdot (1 - e^{-0.329}) + 0.329} \\
  &= \frac{0.329\,T}{2.1 \times 0.280 + 0.329} \\
  &= \frac{0.329\,T}{0.917} \\
  &\approx 0.359\,T
\end{aligned}
\end{equation}
At the operational window $T=300$ (30 s at 100\,ms sampling):
$\Teff\approx 108$.

Substituting into \eqref{eq:convergence_bound} with
$K=3$, $\lambda=10^{-3}$, and $\delta=0.05$:
\begin{equation}
\begin{aligned}
  \norm{\hat{\bm\theta}-\bm\theta^*}_2 
  &\leq \frac{2}{10^{-3}} \sqrt{\frac{8\log(40)}{108}} \\
  &= 2000 \sqrt{\frac{8 \times 3.69}{108}} \\
  &= 2000 \times 0.523 \\
  &\approx 0.52
\end{aligned}
\end{equation}
This theoretical bound is practically informative (given the parameters are $O(1)$) 
and aligns with the observed $0.3$\,pp parameter variation across CV folds.
The long-run variance ratio for the testbed residuals is
$\hat\Sigma_\gamma = \hat\sigma^{-2}\sum_{h=-L}^{L}(1-|h|/(L+1))\hat\gamma_h
\approx 1.83$ at lag truncation $L=\lfloor T^{1/3}\rfloor=6$,
confirming $\hat\Sigma_\gamma/\hat\kappa_4\approx 3.1>1$
and hence indicating that the cumulant correction of Theorem~\ref{thm:finite_sample}
reduces the theoretical KS error bound by a factor of $\approx3$.

The mixing coefficient estimation is performed once on the
training portion of the dataset and held fixed; continuous re-estimation
at each window would likely introduce additional variance without
yielding substantial improvements to the bound at $T=300$.

\subsection{Proof of Proposition~\ref{prop:unique} (Detailed)}
\label{app:unique}

The core of the proof relies on establishing that the strict convexity of the I-S
divergence precludes any alternative function from achieving an identical 
divergence value as \eqref{eq:contention}.
Formally, for any $f_{k_0}\neq c_{k_0}\sigma(U-\tau_{k_0})$
on a set of positive $\mu$-measure:
\begin{align}
  &\int B_\phi(c_{k_0}\|f_{k_0}(U))\,d\mu(U)
  -\int B_\phi(c_{k_0}\|c_{k_0}\sigma(U-\tau_{k_0}))\,d\mu(U)
  \notag\\
  &\geq\int_{\{f_{k_0}\neq c_{k_0}\sigma\}}
    \frac{(f_{k_0}-c_{k_0}\sigma)^2}
         {2\max(f_{k_0},c_{k_0}\sigma)^2}\,d\mu(U)
  > 0,
\end{align}
where the first inequality uses the second-order lower bound
of the I-S divergence:
$B_\phi(a\|b)\geq B_\phi(a\|b_0)+\nabla_b B_\phi(a\|b_0)(b-b_0)
+(b-b_0)^2/(2b_{\max}^2)$ for $b$ between $b_0$ and $b_{\max}$.
The strict positivity of the integral follows directly from
$\mu(\{f_{k_0}\neq c_{k_0}\sigma\})>0$.

\bibliographystyle{IEEEtran}
\bibliography{references}

\end{document}